\documentclass[journal]{IEEEtran}

\usepackage[utf8]{inputenc}
\usepackage{cite}
\usepackage{graphicx}
\usepackage[caption=false]{subfig}
\usepackage{mwe}
\usepackage{lipsum}
\usepackage{xcolor}
\usepackage{pdfpages}

\usepackage{amsmath,amsthm,amssymb}
\usepackage{epstopdf}
\usepackage{booktabs}
\usepackage{graphicx}

\usepackage{cite}
\usepackage{array}
\usepackage{amsmath}
\usepackage{amssymb}
\usepackage{textpos}
\usepackage{float}

\usepackage{enumerate}
\usepackage{lscape}

\usepackage{url}
\usepackage{soul}
\usepackage{color}
\usepackage{threeparttable}
\usepackage{multirow}
\usepackage{hyperref}
\usepackage{amsthm}
\usepackage{caption}
\usepackage{mathtools}

\newtheorem{proposition}{Proposition}
\newtheorem{lemma}{Lemma}
\newtheorem{remark}{Remark}


\ifCLASSINFOpdf

\else
\fi

\begin{document}
\title{Secrecy Analysis of Random MIMO Wireless Networks over $\alpha$-$\mu$ Fading Channels}

\author{Long Kong, \IEEEmembership{Student Member, IEEE}, Satyanarayana Vuppala, \IEEEmembership{Member, IEEE}, and  Georges Kaddoum, \IEEEmembership{Member, IEEE}
\vspace{-2ex}
\thanks{L. Kong and G. Kaddoum are with the Department of Electrical Engineering, École de technologie supérieure (ÉTS), Université du Québec, Montreal (Québec) H3C 1K3, Canada. E-mail: (long.kong.1@ens.etsmtl.ca, georges.kaddoum@etsmtl.ca.)}
\thanks{S. Vuppala is with United Technologies Research Center, Cork, Ireland. Email: vuppalsa@utrc.utc.com}}

\maketitle

\begin{abstract}
In this paper, we investigate the secrecy performance of stochastic MIMO wireless networks over small-scale $\alpha$-$\mu$ fading channels, where both the legitimate receivers and eavesdroppers are distributed with two independent homogeneous Poisson point processes (HPPPs). Specifically, accounting for the presence of non-colluding eavesdroppers, secrecy performance metrics, including the connection outage probability (COP), the probability of non-zero secrecy capacity (PNZ) and ergodic secrecy capacity, are derived regarding the $k$-th nearest/best user cases. The index for the $k$-th nearest user is extracted from the ordering, in terms of the distances between transmitters and receivers, whereas that for the $k$-th best user is based on the combined effects of path-loss and small-scale fading. In particular, the probability density functions (PDFs) and cumulative distribution functions (CDFs) of the composite channel gain, for the $k$-th nearest and best user, are characterized, respectively. Benefiting from these results, closed-form representations of the COP, PNZ and ergodic secrecy capacity are subsequently obtained. Furthermore, a limit on the maximal number of the best-ordered users is also computed, for a given secrecy outage constraint. Finally, numerical results are provided to verify the correctness of our derivations. Additionally, the effects of fading parameters, path-loss exponent, and density ratios are also analyzed. 
\end{abstract}


\IEEEpeerreviewmaketitle
\section{Introduction}
\IEEEPARstart{T}{he} security issue impacting the wireless networks has recently attracted significant attention from the  academic and industrial communities. In this vein, the development of conventional approaches, based on cryptography techniques, faces new challenges, especially in large-scale wireless network, due to its high power consumption and complexity requirements. Alternatively, physical layer security (PLS) appears as an appealing strategy to address such a concern by conversely exploiting the inherent randomness and noisy characteristics of radio channels in order to protect confidential messages from being wiretapped. 
\subsection{Background and Related Works}
The fundamental of the PLS was initially built on the discovery of `\textit{perfect secrecy}' by Shannon \cite{Shan49} and the conceptual finding of degraded `wiretap channel', for the discrete memoryless channel, by Wyner \cite{Wyner75}. Later on, successive efforts were devoted to the generalization of the results in \cite{Wyner75} to additive Gaussian noise channels \cite{Leung78}, broadcast channels \cite{Csiszar78}, fading channels \cite{4529264,4626059,7094262,7374839,7856980,8424118}, multiple-input multiple-output (MIMO) communications \cite{6549315,7024178,7881216}, cooperative networks \cite{7299681}, cellular networks \cite{7931549,8119548} among other topics. 

A common shortage of the aforesaid works \cite{Leung78,Csiszar78,4529264,4626059,7094262,7374839,7856980,7881216,7299681,8424118}, based on the point-to-point communication links, lies in the uncertainty of users' spatial locations. Strictly speaking, users' spatial locations undoubtedly play a crucial role when investigating the large scale fading in random networks. The pioneer works, led by Haenggi \cite{4675736,4595044}, and where users distributed randomly based on stochastic geometry, was modeled as the Poisson
point process (PPP). Specifically, it is worthy to mention that the concept of `secrecy graph' was firstly proposed to study the secrecy connectivity metric, and subsequently the maximum secrecy rate \cite{5995291} when colluding eavesdroppers are considered.    

More recently, in \cite{6754116,6823670,6840982,7136149,7742344}, the authors considered the two-dimensional random wireless network under Rayleigh, composite fading and Nakagami-\textit{m} fading channels, where both the legitimate receivers and eavesdroppers are drawn from two independent homogeneous PPPs (HPPPs). Authors in \cite{6754116} studied the secure MIMO transmission subjected to Rayleigh fading. Zheng \textit{et al.} in \cite{6840982} analyzed the transmission secrecy outage probability for multiple-input and single-output (MISO) systems, and proposed the concept of `\textit{security region} (SR)', which is a geometry region, defined as the legitimate receiver's locations having a certain guaranteed level of secrecy. Differently, Satyanarayana \textit{et al.} proposed another SR \footnotemark[1]\footnotetext[1]{Within the security region, all users can achieve high secrecy gains.} \cite{7510876,7942069}, which is defined as the region where the set of ordered nodes can safely communicate with typical destination, for a given secrecy outage constraint.

Motivated by those references, it is thus of tremendous significance to study how many legitimate users are located within the coverage of the transmitter (i.e., base station), in the presence of unknown number of eavesdroppers. Most of the existing work can be summarized in terms of the ordering policy, namely the $k$-th legitimate user, either based on the distances between transmitters and users, or the instantaneous received composite channel gain. Moving in this direction, it is reported that limited studies are seen on the secrecy assessment of the $k$-th legitimate receiver. Specifically, the result disclosed in \cite{6823670} is merely restricted to the mathematical treatment of the secrecy outage probability of the $k$-th nearest receiver (i.e., the index is from the ordering based on the distance between the source and the destination). In contrast, the results unveiled in \cite{7136149,7742344,7510876,7942069} are constrained to the $k$-th best receiver\footnotemark[2]\footnotetext[2]{It is worth mentioning that the $k$-th best user is the one with the $k$-th maximal received signal out of $K$ users.} (i.e., the index is according to the array of the composite channel gains). It is reported that \cite{7888950} investigated the secrecy issue over Rayleigh fading channels, while considering both ordering policies without offering any SR. On the other hand, the introduced $k$-th nearest or best receiver is applicable to vehicular networks. The $k$-th best user can be considered as any potential vehicle receiving the \emph{$k$-th maximum path gain} from a source vehicle. One can construct the security region by selecting all the best nodes instead of random users. Selecting the best users to coordinate among each other can further improve the security of the network.

Outstandingly, the aforesaid studies merely focus on the secrecy analysis, influenced by the colluding/non-colluding eavesdroppers but have not taken the more general fading model, namely, $\alpha$-$\mu$ fading channel, into consideration. The $\alpha$-$\mu$ distribution was first proposed by Yacoub in \cite{4067122} to model the small scale variation of fading signal under line-of-sight conditions \cite{5200462,7243294}. It is physically described with two key fading parameters, i.e., non-linearity of the propagation medium $\alpha$ and the clustering of the multipath waves $\mu$. The advantage of these two factors is regarded as a useful tool to vividly depict the inhomogeneous surroundings compared with other existing fading models, such as Rayleigh, Nakagami-$m$. Most of them are based on the unrealistic assumption of homogeneous (scattering) environment. Fortunately, the $\alpha$-$\mu$ fading model was later on found valid and feasible in many realistic situations \cite{5425619,5419050,5508963,5728940,5782205,6203563,6516883}, including the vehicle-to-vehicle (V2V) communication networks and wireless body area networks (WBAN). In addition, the $\alpha$-$\mu$ distribution is flexible and mathematical tractable, since it can be extended to Rayleigh, Nakagami-\textit{m} and Weibull fading by simply attributing the fading parameters $\alpha$ and $\mu$ to selected values. For example, choosing $\alpha = 2$ and $\mu =1$ will reduce it to Rayleigh fading, while choosing $\alpha=2$ and $\mu=m$ will make it correspond to Nakagami-$m$ fading. 

To the best knowledge of the authors, in \cite{7094262,7374839,8424118}, the authors derived the probability of non-zero secrecy capacity and secrecy outage probability of point-to-point communication over $\alpha$-$\mu$ fading channels. Lei \textit{et al.} \cite{7856980} later on studied the average secrecy capacity of $\alpha$-$\mu$ wiretap fading channels. The importance of evaluating the aforementioned two metrics is based on the behavior of the eavesdroppers. If they are active, meaning that it is possible to have their channel state information (CSI) at the transmitter, the probability of non-zero secrecy capacity and the secrecy outage probability are crucial. If they are passive, average secrecy capacity is therefore a key benchmark. With respect to the random single-input and single-output (SISO) wireless networks, the authors in \cite{7510876,7942069} and \cite{7247765} correspondingly investigated the secrecy outage probability and the ergodic secrecy capacity in terms of the $k$-th best user, respectively. Apart from the literature \cite{7094262,7374839,8424118,7856980,7510876,7942069,7247765,8354927,ICCWORKSHOPlong2018}, efforts to explore the secrecy evaluation of random MIMO wireless networks over $\alpha$-$\mu$ fading channels are rarely witnessed.
\subsection{Contribution and Organization}
Consequently, the essence of this paper is the exploration of the $k$-th legitimate user's secrecy performance over $\alpha$-$\mu$ fading channel in typical random wireless networks. 

In this paper, we consider a stochastic MIMO wireless system, in the presence of two types of receivers, namely, legitimate users and eavesdroppers. They are assumed to be distributed with two independent HPPPs. The conventional space-time transmission (STT) is considered \cite{7024178}. All receivers have access to perfect channel state information (CSI), which are all subjected to quasi-static $\alpha$-$\mu$ fading. Since Wyner had concluded that perfect secrecy can be assured only if legitimate links have higher transmission rate, compared to wiretap links, the pursuit of outage-based secrecy performance analysis is considered reasonable and feasible when a fixed data transmission scheme is adopted for such quasi-static fading channels, as indicated in \cite{7464352,7742344}. In \cite{7247765}, the secure connection probability of the $k$-th legitimate receiver to the transmitter was studied, as well as the ergodic secrecy capacity. 

To this end, the connection outage probability (COP), the probability of non-zero secrecy capacity (PNZ) and the ergodic secrecy capacity, in terms of the $k$-th nearest and best legitimate receivers, are taken into consideration.

Since the concept of the $k$-th best user can be regarded as a security region, it is crucial to identify the $k^*$ best users out of $K$ users that can communicate securely with the transmitter in such region. In this work, we identify a zone (i.e., a limited number of legitimate users) comprising of such ordered $k^*$ best users, for a given secrecy constraint. 

The contributions of this paper are multifold, which can be pointed out as: 
\begin{itemize}
\item The probability of density functions (PDFs) and cumulative distribution function (CDFs) of the composite channel gain for the $k$-th nearest and best user are derived, respectively. This is essentially important for formulating the secrecy metrics, including the connection outage probability, the probability of non-zero secrecy capacity and ergodic secrecy capacity.
\item Unlike the model studied in \cite{7856980}, which considered the point-to-point system and a single eavesdropper, we study the secrecy capacity of random networks with multiple legitimate receivers and eavesdroppers. The exact closed form expressions of the COP, PNZ and ergodic secrecy capacity of the $k$-th legitimate user are derived.
\item Motivated by the PNZ of the $k$-th best receiver, a limit on the maximal number of the best-ordered receivers is calculated thereafter respecting a given secrecy outage probability. In other words, this limit eventually provides a security region concept, henceforth, all the system parameters are looked upon, based on this concept, giving a better insight into the secrecy capacity regions of random wireless networks. 
\item The accuracy of our derivations are successfully validated by Monte-Carlo simulation. Numerical outcomes are presented to indicate the effect of the path-loss exponent, densities of the users and fading parameters.  
\end{itemize}
The insights obtained from the outcomes of this paper, regarding the crucial parameters of the secrecy performance, inspire researchers and vehicle wireless communication engineers to quickly evaluate system performance and optimize available parameters when confronting various security risks.

The rest of this paper is organized as follows: system model and problem formulation are depicted in Sections \ref{Sec_systemmodel} and \ref{Sec_ProblemFormulation}, respectively. The COP, PNZ and the ergodic secrecy capacity are derived in Section \ref{Sec_noncoll}. Numerical results and discussions are then presented in Section \ref{SimulationResults} and followed by Section \ref{Concluding} with concluding remarks. Notations and symbols used in this paper are shown in Table. \ref{T1}. 
\begin{table}[!t]
\newcommand{\tabincell}[2]{\begin{tabular}{@{}#1@{}}#2\end{tabular}}
  \begin{center}
  \renewcommand{\arraystretch}{1.3}
    \caption{Notations and symbols}
    \label{T1}
    \begin{tabular}{ll}
      \toprule
       Notations & Description \\
      \midrule
      $[x]^+$& $[x]^+=\max(x,0)$\\
      $\mathbb{N}$& positive integer \\
      $\mathbb{E}$& expectation operator\\
      i.i.d & identical independent distributed\\
      $R_t$ & transmission rate  \\
      $d$ & dimensions of the network \\
      $r$&distance from the origin to the receiver \\
      $\upsilon$ & path-loss exponent \\
      $f_X$ & PDF  of $X$  \\
      $F_X$ & CDF of $X$\\
      $c_d$ & $\pi^{d/2}/\Gamma(1+d/2)$\\
      $\delta$ & $d/\upsilon$\\
      $\Psi$ & path-loss process before fading \\
      $\Xi_k$ & path-loss process with fading for legitimate users\\
      $\Xi_e$ & path-loss process with fading for eavesdroppers\\
      $\lambda_b$ & density for legitimate receivers\\
      $\lambda_e$ & density for eavesdroppers\\
      $\Gamma(a)$& \tabincell{l}{$\Gamma(a)=\int_0^\infty t^{a-1}e^{-t}dt$\\Gamma function \cite[eq. (8.310.1)]{gradshteyn2014table}}\\
      $\gamma(a,x)$& \tabincell{l}{$\gamma(a,x)=\int_0^xt^{a-1}e^{-t}dt $\\lower incomplete gamma function \cite[eq. (8.350.1)]{gradshteyn2014table}}\\
      $\Gamma(a,x)$& \tabincell{l}{$\Gamma(a,x)=\int_x^\infty t^{a-1}e^{-t}dt$\\ upper incomplete gamma function\cite[eq. (8.350.2)]{gradshteyn2014table}}\\
      $H_{m,n}^{p,q}[\cdot]$ & Fox's $H$-function \cite[eq. (1.1.1)]{mathai1978h} \\
      \bottomrule
    \end{tabular}
  \end{center}
  \vspace{-0.3cm}
\end{table}  
\section{System Model} \label{Sec_systemmodel}
In this paper, a random wireless network, displayed in Fig. \ref{SysModel} in an unbound Euclidean space of dimension $d$ is under consideration. The typical transmitter is located at the origin, who has $N_a$ ($N_a \ge 1$) antennas, and two types of receivers, namely the legitimate receivers and eavesdroppers with $N_b$ ($N_b \ge 1$), $N_e$ ($N_e \ge 1$) antennas, respectively. The locations of these receivers are drawn from two independent HPPPs. Their location sets are separately denoted by $\Phi_b(\lambda_b)$ and $\Phi_e(\lambda_e)$ with corresponding densities $\lambda_b$ and $\lambda_e$ \cite{6754116,7136149}. In such a network configuration, it is assumed that the communication links undergo a path-loss characterized by the exponent $\upsilon$ and $\alpha $-$\mu $ fading. 
\begin{figure}[!t]
\centering{\includegraphics[width=3in]{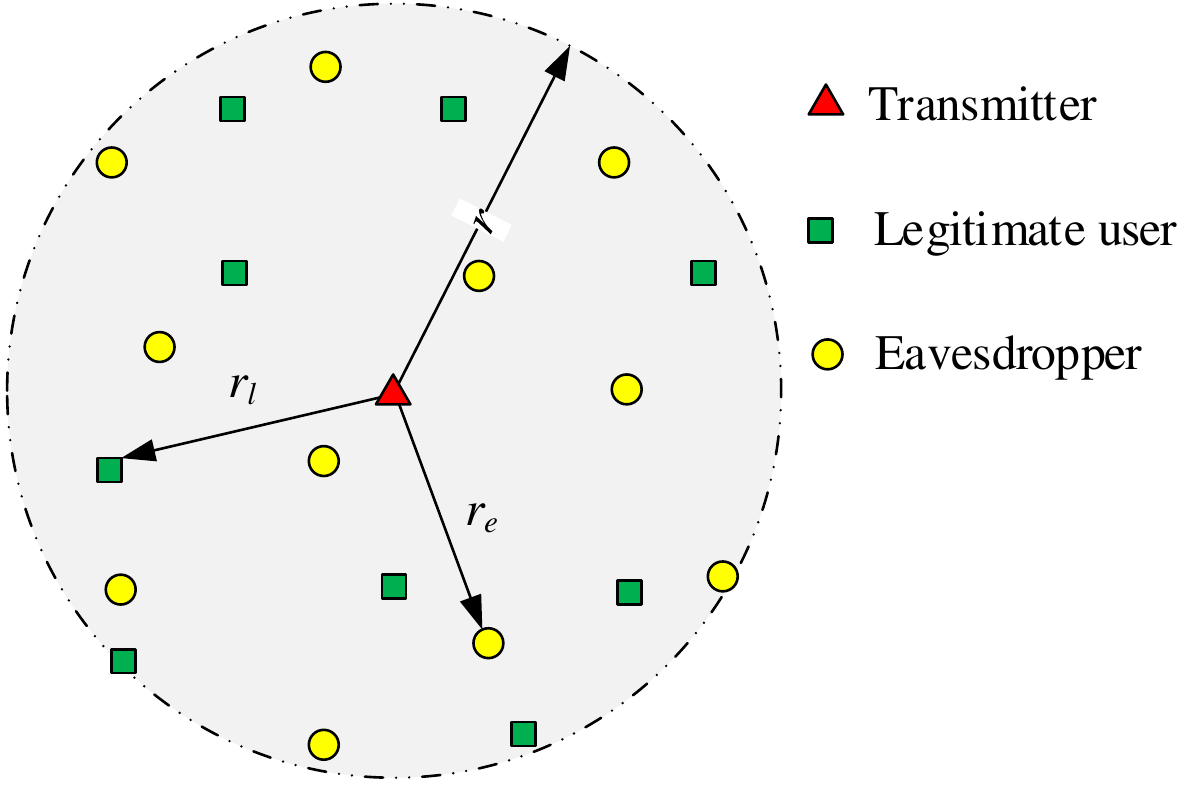}}
\caption{ \small{A 2-dimensional stochastic MIMO wireless network with independently HPPP distributed legitimate receivers and eavesdroppers.} }
\label{SysModel}
\end{figure}

Consider a transmitter that intends to send private messages to a legitimate user in the presence of eavesdroppers located at some unknown distances $r_e$. In such a stochastic MIMO wireless system, the conventional STT scheme is considered at the transmitter and receivers \cite{7024178}, then the instantaneous received signal-to-noise ratios (SNRs) at a legitimate user, $\gamma_b$, and an eavesdropper, $\gamma_e$, would be expressed as \cite[eq.(1)]{7024178}
\begin{subequations}
\begin{equation}
\gamma_b = \frac{P \sum \limits_{n_a=1}^{N_a} \sum \limits_{n_k=1}^{N_b} g_{n_a,n_k} }{r_{l}^{\upsilon}\sigma_k^2}=\eta_k\frac{g_k}{r_{l}^{\upsilon}},
\end{equation}
\begin{equation}
\gamma_e = \frac{P\sum \limits_{n_a=1}^{N_a} \sum \limits_{n_e=1}^{N_e }g_{n_a,n_e}}{r_e^{\upsilon}\sigma_e^2}=\eta_e\frac{g_e}{r_e^{\upsilon}} ,
\end{equation}
\end{subequations} 
where $\eta_i = \frac{P}{\sigma_i^2}$, $g_{n_a,n_i}=\vert h_{n_a,n_i} \vert^{2}, i \in \lbrace k,e\rbrace $, denote the instantaneous channel power gain with unit mean. $P$ denotes the transmission power and the terms $\sigma_i$ denote the noise power at the legitimate and eavesdropping receivers, respectively. So herein, $r_{l}$ and $h_{n_a,n_k}$ are the distance and fading envelope from the transmitter to the $k$-th legitimate receiver, respectively. Similarly, $r_e$ and $h_{n_a,n_e}$ are the distance and fading envelope from the transmitter to the eavesdropper, respectively. Here, $h_{n_a,n_i}$ are modeled by $\alpha$-$\mu$ fading with an arbitrary fading parameter $\alpha_i >0$ and an inverse normalized variance of $h_i^{\alpha_i}$ denoted as $\mu_i$. 

Since STT scheme is used, $g_i$ is obviously the sum of all the receivers' channel gain. Recalling the results obtained in \cite{Costa_TWC_08}, the exact PDF and CDF of $g_i$ are too complex due to the convolution of $M$ PDFs of each eavesdropper's channel gain when developing the secrecy performance. Thanks to the highly tight approximation method provided therein, it is deduced therein that the PDF of $g_i$ is given as the following form with parameters ($\alpha_{i}$, $\mu_{i}$, $\Omega_i$)\footnotemark[3]\footnotetext[3]{The method of obtaining all these three parameters is suggested to refer to \cite{Costa_TWC_08}.}, 
\begin{equation} \label{PDFgi}
\begin{split}
f_{g_i}(x) & \approx \frac{{\alpha _i x^{\frac{{\alpha _i \mu _i }}{2} - 1} }}{{2\Omega _i^{\frac{{\alpha _i \mu _i }}{2}} \Gamma \left( {\mu _i } \right)}}\exp \Bigg( { - \left( {\frac{x}{{\Omega _i }}} \right)^{\frac{{\alpha _i }}{2}} } \Bigg) \\
& = \epsilon_i H_{0,1}^{1,0} \left[ { \theta_i x \left| {\begin{array}{*{20}c}
    {-}  \\
   {(\mu_i - \frac{2}{\alpha_i}, \frac{2}{\alpha_i })}  \\
\end{array}} \right.} \right],
\end{split}
\end{equation} 
where $\Omega _i =\frac{\Gamma(\mu_i)}{\Gamma\left(\mu_i+\frac{2}{\alpha_i}\right)} $, $\epsilon_i =\frac{1}{\Omega _i \Gamma(\mu_i)} $, and $\theta_i= \frac{1}{\Omega _i}$. After integrating (\ref{PDFgi}), the CDF of $ g_i $ is given by
\begin{equation} \label{CDFgi}
\begin{split}
F_{g_i}(x)& = \frac{{ \gamma\left(\mu_i, \left( \frac{x}{\Omega_i}\right)^{\frac{\alpha_i}{2}}\right) }}{\Gamma \left( {\mu _i } \right)} \\
& = 1 - \frac{\epsilon_i}{\theta_i} H_{1,2}^{2,0} \left[ { \theta_i x \left| {\begin{array}{*{20}c}
    {(1,1)}  \\
   {(0,1),(\mu_i, \frac{2}{\alpha_i })}  \\
\end{array}} \right.} \right].
\end{split}
\end{equation}
\vspace{-0.3cm}
\section{Problem Formulation} \label{Sec_ProblemFormulation}
\subsection{User Association}
\subsubsection{The nearest user} In this case, all the receivers are ordered according to their distance from the transmitter. Let $\{ r_k\}$ be a random set of legitimate receivers in ascending order of the distances from the receiver to the transmitter (i.e., $\vert r_1\vert < \vert r_2\vert<\vert r_3\vert< \cdots$ ). Letting $Z = \frac{g_k}{r_k^{\upsilon}}$, the PDF and CDF of the composite channel gain are respectively given in the following \textbf{Lemma}. 
\begin{lemma} \label{Lemma1}
{\rm The PDF and CDF of the composite channel gain for the $k$-th nearest legitimate user are given by (\ref{Pdf_nearest}) and (\ref{CDF_nearest}) in terms of the Fox's $H$-function\footnotemark[4]\footnotetext[4]{The numerical evaluation of Fox's $H$-function for MATLAB implementations is to use the method given in \cite[Table. II]{peppas2012simple}.}, respectively. 
\begin{subequations}
\begin{equation} \label{Pdf_nearest}
\begin{split}
&f_{\frac{g_k}{r_k^{\upsilon}}}(z)  = \frac{\epsilon_k}{A_k^{\frac{1}{\delta}} \Gamma(k)  } H_{1,1}^{1,1} \left[ { \frac{\theta_k z}{ A_k^{\frac{1}{\delta}}}  \left| {\begin{array}{*{20}c}
    {(1 - k - \frac{1}{\delta},\frac{1}{\delta})}  \\
   {(\mu_k - \frac{2}{\alpha_k }, \frac{2}{\alpha_k })}  \\
\end{array}} \right.} \right],
\end{split}
\end{equation}
\begin{equation} \label{CDF_nearest}
\begin{split}
\hspace{-1ex}F_{\frac{g_k}{r_k^{\upsilon}}}\left(z\right) =  1 - \frac{ \epsilon_k}{{\theta_k\Gamma(k)}}  H_{2,2}^{2,1} \left[ { \frac{\theta_k z}{ A_k^{\frac{1}{\delta}}}  \left| {\begin{array}{*{20}c}
   {(1-k,\frac{1}{\delta}),(1,1)}  \\
   {(0,1),(\mu_k, \frac{2}{\alpha_k})}   \\
\end{array}} \right.} \hspace{-1ex} \right],
\end{split}
\end{equation}
\end{subequations}
where $A_k = \pi \lambda_b$.}
\end{lemma}
\begin{proof}
See Appendix \ref{Appen_A} and \ref{Appen_B}, respectively.
\end{proof}
Similarly, the PDF and CDF for the $k$-th nearest eavesdropper can be obtained with parameters $A_e = \pi \lambda_e$.
\subsubsection{The best user}
Unlike the nearest user, the $k$-th best user describes the ordering of the recerivers according to the received SNR function of the combination of the path-loss and small-scale fading. Letting $\Xi_k=\{\xi_k \triangleq r_k^{\upsilon}/g_k, k \in \mathbb{N}\} $ be the path-loss process with small-scale fading. It is reported in \cite{4675736} that $\Xi_k$ is also a PPP with the intensity function $\lambda_{\Xi_k}$. For the $k$-th best user, we have $\vert \xi_1\vert < \vert \xi_2\vert < \vert \xi_3\vert< \cdots$, since $\xi_k$ takes the inverse shape of the composite channel gain.  
\begin{lemma} 
{\rm Given the path-loss process with the $\alpha$-$\mu$ fading, the intensity of $\Xi_k$ is given by
\begin{equation}
\lambda_{\Xi_k}=A_{b0} x^{\delta-1},
\end{equation}
where $A_{b0} = \frac{\lambda_b c_d \delta \Omega_k^\delta \Gamma\left(\mu_k + \frac{2\delta}{\alpha_k}\right)}{\Gamma(\mu_k)} $.}
\end{lemma} 
\begin{proof}
See Appendix \ref{Appen_C}.
\end{proof}
Similarly, with regard to eavesdroppers, we have $\Xi_e = \{ r_e^\upsilon/g_e, e \in \mathbb{N}\}$, $\lambda_{\Xi_e} = A_{e0} y^{\delta -1}$, $A_{e0} = \frac{\lambda_e c_d \delta \Omega_e^\delta \Gamma\left(\mu_e + \frac{2\delta}{\alpha_e}\right)}{\Gamma(\mu_e)}$. 

Let $\frac{1}{\xi_k} = Z$, then the PDF and CDF of $\frac{1}{\xi_k} $ are provided in the following Lemma.
\begin{lemma}
{\rm The PDF and CDF of the composite channel gain for the $k$-th best user are 
\begin{subequations}
\begin{equation} \label{PDF_best}
f_{\frac{1}{\xi_k}}(z)=\exp\left(-A_{b1}z^{-\delta}\right)\frac{\delta(A_{b1}z^{-\delta})^k}{z^{-1}\Gamma(k)}.
\end{equation}
\begin{equation} \label{CDF_best}
F_{\frac{1}{\xi_k}}(z) = \frac{\Gamma(k,A_{b1}z^{-\delta})}{\Gamma(k)},
\end{equation}
\end{subequations}
where $A_{b1}=A_{b0}/\delta$.}
\end{lemma}
\begin{proof} 
See Appendix \ref{Appen_D}.
\end{proof}
\begin{figure}[!t]
\centering{\includegraphics[width=\columnwidth]{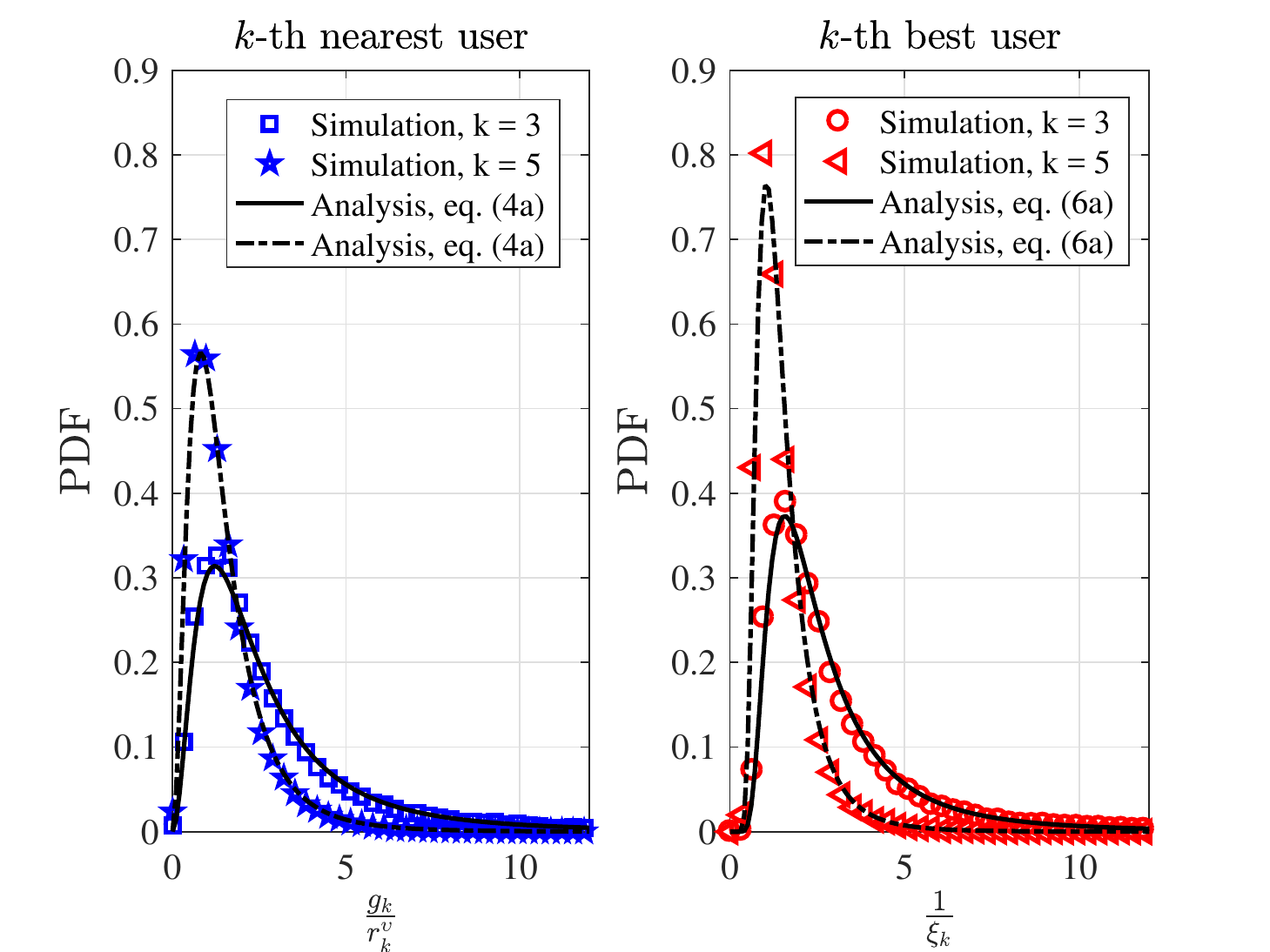}}
\caption{ \small The PDFs for the $k$-th best and nearest users when $\alpha_k = 2$, $\mu_k = 3$, $\eta_k = 0$ dB, $d = \upsilon = 2$, $\lambda_b = 2$, $N_a = N_b = 1$.}
\label{Pdf_best_near}
\end{figure}
 As shown in Fig. \ref{Pdf_best_near}, the PDFs for the $k$-th nearest and the $k$-th best legitimate user are respectively demonstrated, it is observed that our analysis are successfully validated by simulation results.
\subsection{Secrecy Metrics}
\subsubsection{Connection outage probability}
Connection outage probability is defined as the event in which the legitimate receiver cannot successfully decode the transmitted messages. This happens when the main channel capacity falls below the actual transmission rate $R_t$. It is mathematically defined as
\begin{equation}
\mathcal{P}_{co}(R_t)  =\mathcal{P}r\Bigg(\log_2\left(1+\frac{\eta_k g_k}{r_l^{\upsilon}}\right) <R_t \Bigg).
\end{equation}
\subsubsection{Probability of non-zero secrecy capacity} 
The secrecy capacity of the aforementioned system model under the assumption that eavesdroppers do not collude, is \cite{7136149}
\begin{equation} \label{Csk}
C_{s:k}= \left [ \log_2\left(1+\frac{\eta_k g_k}{r_l^{\upsilon}}\right)-\log_2\left(1+\frac{\eta_e g_e}{r_e^{\upsilon}}\right) \right]^+.
\end{equation}
When the wiretap channel capacity is less than the main channel capacity, the eavesdroppers are incapable of successfully decoding the transmitted messages. The probability of the occurrence for this event is called as the probability of non-zero secrecy capacity. Mathematically from (\ref{Csk}), the probability of non-zero secrecy capacity is defined as
\begin{equation}
\mathcal{P}_{nz}  = \mathcal{P}r \left( \frac{\eta_k g_k}{r_l^{\upsilon}} > \frac{\eta_e g_e}{r_e^{\upsilon}} \right).
\end{equation}

\subsubsection{Ergodic secrecy capacity} 
In line with \cite{5524086,6549315,6805139,7247765,7302060,LongIWCMC2018}, the ergodic secrecy capacity is obtained as follows
\begin{equation} \label{Cs_ergodic}
C_{s:k} = \left[\underbrace{\mathbb{E}\left[\log_2 \left(1 + \frac{\eta_k g_k}{r_l^\upsilon}\right)\right]}_{R_{k}^M }- \underbrace{\mathbb{E}\left[\log_2 \left(1 + \frac{\eta_e g_e}{r_e^\upsilon}\right)\right]}_{R_{k}^W} \right]^+,
\end{equation}
where $R_{k}^M$ and $R_{k}^W$ are the ergodic capacity of the transmitter to the $k$-th legitimate receiver and the $k$-th eavesdropper, respectively. 

\section{Performance Characterization} \label{Sec_noncoll}
By using the PDFs and CDFs of the composite channel gain for the $k$-th nearest/best user, we study the COP, PNZ, and ergodic secrecy capacity, respectively. 
\subsection{Performance Characterization of the COP}
\subsubsection{Connection outage probability for the $k$-th nearest receiver} 
From the definition, the COP for the $k$-th nearest legitimate receiver is mathematically expressed as 
\begin{equation} \label{pco_nearest}
\begin{split}
\mathcal{P}_{co,N}(R_t) & =\mathcal{P}r\Bigg(\log_2\left(1+\frac{\eta_k g_k}{r_k^{\upsilon}}\right) <R_t \Bigg)\\
&=\mathcal{P}r\left(\frac{g_k}{r_k^{\upsilon}} < \frac{2^{R_t}-1}{\eta_k} \right).\\
\end{split}
\end{equation}
Notably, $\mathcal{P}_{co,N}(R_t)$ can be assessed from the PDF of the $k$-th legitimate receiver's channel gain. For the ease of notations, we set $\Delta = \frac{2^{R_t}-1}{\eta_k}$.
\begin{proposition}
{\rm The COP of the $k$-th nearest legitimate receiver is given as}
\begin{equation}
\mathcal{P}_{co,N}(R_t) =  F_{\frac{g_k}{r_k^{\upsilon}}}\left(\Delta\right).
\end{equation}
\end{proposition}
\begin{proof}
Substituting (\ref{CDF_nearest}) into (\ref{pco_nearest}), the proof is achieved.
\end{proof}
\subsubsection{Connection outage probability for the $k$-th best receiver} 
Similarly, the COP for the $k$-th best receiver is given  
\begin{equation} \label{pco_best}
\begin{split}
\mathcal{P}_{co,B}(R_t) & =\mathcal{P}r\Bigg(\log_2\left(1+\frac{\eta_k}{\xi_k }\right) <R_t \Bigg)\\
&=1-\mathcal{P}r\left(\xi_k < \frac{1}{\Delta }\right).\\
\end{split}
\end{equation}
Based on (\ref{pco_best}), it is becoming apparent that the COP for the $k$-th best receiver is termed as $F_{\xi_k}$.

\begin{proposition}
{\rm The COP of the $k$-th best legitimate receiver takes the following shape}
\begin{equation} \label{Pco_case2}
\mathcal{P}_{co,B}(R_t) = \frac{\Gamma\left( k, A_{b1} \Delta^{-\delta}\right)}{\Gamma(k)}.
\end{equation}
\end{proposition}
\begin{proof}
Substituting (\ref{CDF_best}) into (\ref{pco_best}), the proof is completed.
\end{proof}

\subsection{Performance Characterization of the PNZ} \label{Sec_coll}
In this section, the PNZs, with respect to the $k$-th nearest and best legitimate receiver, are well investigated.

As seen from (\ref{Csk}) for the non-colluding eavesdroppers, the non-zero secrecy capacity for the $k$-th legitimate receiver is mathematically guaranteed with the probability given for the following four scenarios:
\begin{itemize}
\item[$\bullet$] case 1): the $k$-th nearest legitimate receiver in the presence of the 1st nearest eavesdropper\footnotemark[5]\footnotetext[5]{The nearest eavesdropper is the one closest to the legitimate receiver.};
\begin{equation} \label{PSPC_nearest}
\begin{split}
\hspace{-5ex}\mathcal{P}_{nz,NN}& = \mathcal{P}r\left(\frac{\eta_k g_k }{r_k^{\upsilon}}>\frac{\eta_e g_e}{r_e^{\upsilon} }\right)=\mathcal{P}r\left( \frac{g_e}{r_e^{\upsilon}}\frac{r_k^{\upsilon}}{g_k} < \frac{\eta_k}{\eta_e}\right) \\
& = \int_0^\infty F_{\frac{ g_e}{r_e^\upsilon}}(\varpi y) f_{\frac{ g_k}{r_k^\upsilon}} \left( y \right) dy .
\end{split}
\end{equation}
\item[$\bullet$] case 2): the $k$-th best legitimate receiver in the presence of the 1st best eavesdropper\footnotemark[6]\footnotetext[6]{The best eavesdropper is supposed to be the one with the smallest $\xi_e$.};
\begin{equation} \label{PSPC_best}
\begin{split}
\hspace{-2ex}\mathcal{P}_{nz,BB} &= \mathcal{P}r\left(\frac{ \eta_k}{\xi_k }>\frac{\eta_e}{\xi_e }\right)=1 - \mathcal{P}r\left( \frac{\xi_e}{\xi_k} < \frac{1}{\varpi} \right) \\
&= 1-\int_0^\infty F_{\xi_e}\left(\frac{y}{\varpi} \right) f_{\xi_k} (y) dy.
\end{split}
\end{equation}
\item[$\bullet$] case 3): the $k$-th nearest legitimate receiver in the presence of the 1st best eavesdropper;
\begin{equation} \label{Pnz_NB}
\begin{split}
\hspace{-1.5ex}\mathcal{P}_{nz,NB}& = Pr \left( \frac{\eta_k g_k}{r_k^{\upsilon}} > \frac{\eta_e}{\xi_e} \right)  = 1- Pr \left( \frac{g_k}{r_k^{\upsilon}} \xi_e < \frac{1}{\varpi}  \right)\\
& = 1 - \int_0^\infty F_{\frac{g_k}{r_k^{\upsilon}}}\left( \frac{1}{\varpi y} \right)f_{\xi_e}(y) dy.
\end{split}
\end{equation}
\item[$\bullet$] case 4): the $k$-th best legitimate receiver in the presence of the 1st nearest eavesdropper
\begin{equation}  \label{Pnz_BN}
\begin{split}
\hspace{-2ex}\mathcal{P}_{nz,BN}& = Pr \left(\frac{\eta_k}{\xi_k} > \frac{\eta_e g_e}{r_e^{\upsilon}}  \right)  = Pr \left( \frac{g_e}{r_e^{\upsilon}} \xi_k < \varpi  \right)\\
& = \int_0^\infty F_{\frac{g_e}{r_e^{\upsilon}}}\left( \frac{\varpi}{y}\right)f_{\xi_k}(y) dy.
\end{split}
\end{equation}
\end{itemize}
\subsubsection{The $k$-th nearest receiver \& the 1st nearest eavesdropper}
\begin{proposition} \label{P1}
{\rm The PNZ of the $k$-th nearest legitimate receiver in the presence of the 1st nearest eavesdropper can be calculated from }
\end{proposition}
\begin{equation} \label{Pspc_nz1_nearest}
\begin{split}
 &\hspace{-2ex}\mathcal{P}_{nz,NN} = 1 - \frac{\epsilon_k \epsilon_e }{\theta_k \theta_e \Gamma(k)} \\
 &\hspace{-2ex}\times H_{3,3}^{3,2} \left[ { \frac{\theta_e \varpi}{\theta_k} \left( \frac{A_k}{A_e} \right)^{\frac{1}{\delta} }\left|  {\begin{array}{*{20}c}
    {(1,1),(1 - \mu_k , \frac{2}{\alpha_k }),(0,\frac{1}{\delta})}  \\
   {(0,1),(\mu_e , \frac{2}{\alpha_e}),(k,\frac{1}{\delta})}  \\
\end{array}} \right.} \right].
\end{split}
\end{equation}
\begin{proof}
See Appendix \ref{Appen_E}.
\end{proof}
\subsubsection{The $k$-th best receiver \& the 1st best eavesdropper}
\begin{proposition} \label{P2}
{\rm The PNZ of the $k$-th best legitimate receiver in the presence of the 1st best eavesdropper is given as}
\begin{equation} \label{PSPC_best_final}
\mathcal{P}_{nz,BB}=  \left(\frac{A_{b1}}{A_{b1}+A_{e1}\varpi^{-\delta}} \right)^k. 
\end{equation}
\end{proposition}
\begin{proof} 
Motivated by (\ref{CDF_best_xi}), for the best eavesdropper, the CDF of $\xi_e$ is given by
\begin{equation} \label{CDF_xie}
\begin{split}
F_{\xi_e}(x) &= \gamma\left( 1, A_{e1} x^\delta\right) = 1- \exp(-A_{e1}x^\delta),
\end{split}
\end{equation}
where $A_{e1} = A_{e0}/\delta$.

After plugging (\ref{CDF_xie}) and (\ref{PDF_best_xi}) into (\ref{PSPC_best}), it yields
\begin{equation}
\begin{split}
&\hspace{-1.5ex}\mathcal{P}_{nz,BB} = 1-\mathcal{P}r\left( \frac{\xi_e}{\xi_k} <\frac{1}{ \varpi}\right)\\
&\hspace{-1.5ex}= 1- \int_0^\infty F_{\xi_e}\left( \frac{ y}{\varpi } \right) f_{\xi_k}(y)dy\\
&\hspace{-1.5ex} = \int_0^\infty \exp \left( - A_{e1} \left( \frac{ y}{\varpi } \right)^\delta\right)\exp(-A_{b1}y^\delta)\frac{\delta(A_{b1}y^\delta)^k}{y\Gamma(k)}dy\\
&\hspace{-1.5ex} \mathop=^{(a)}\frac{\delta A_{b1}^k}{\Gamma(k)}\int_0^\infty \exp\left(-(A_{b1}+ A_{e1}\varpi^{-\delta})y^\delta\right)y^{\delta k -1}dy\\
&\hspace{-1.5ex} = \left(\frac{A_{b1}}{A_{b1}+A_{e1}\varpi^{-\delta}} \right)^k,
\end{split}
\end{equation}
where $(a)$ follows from \cite[Eq. (3.351.3)]{gradshteyn2014table}. 
\end{proof}

In the following lemma we characterize a limit on the $k$-th best receiver. From Proposition \ref{P2}, one can obtain the maximum possible $k$-th index for a given probability constraint, $\tau = 1- \mathcal{P}_{nz,BB}$.
\begin{lemma} 
{\rm The maximum number of ordered best intended receivers that can securely communicate with the source in the presence of the best eavesdropper is given as}
\begin{equation} \label{EQ_bestPathGainEvek}
k^*  =  \log_{\frac{A_{b1}}{A_{b1}+A_{e1}\varpi^{-\delta}}}  \left(\tau \right). 
\end{equation}
\end{lemma}
\begin{proof}
The proof directly follows from Proposition \ref{P2}.
\end{proof}
\subsubsection{The $k$-th nearest receiver \& the 1st best eavesdropper}

\begin{proposition} \label{P3}
{\rm The PNZ of the $k$-th nearest legitimate receiver in the presence of the 1st best non-colluding eavesdropper is given by } 
\begin{equation} \label{Pspc_nearest_best}
\begin{split}
&\hspace{-2ex}\mathcal{P}_{nz,NB}=  \frac{ \epsilon_k}{\theta_k  \Gamma(k)} \\ 
&\times H_{3,2}^{1,3}  \left[ { \frac{\varpi }{\theta_k} \left( \frac{A_k}{A_{e1}} \right)^{\frac{1}{\delta}}\left| {\begin{array}{*{20}c}
   {(1,1),(1-\mu_k,\frac{2}{\alpha_k}),(0,\frac{1}{\delta})}  \\
   {(k, \frac{1}{\delta}),(0,1)}   \\
\end{array}} \right.}\hspace{-1.5ex} \right]. 
\end{split}
\end{equation}
\end{proposition}
\begin{proof} 
See Appendix \ref{Appen_F}.
\end{proof}
\subsubsection{The $k$-th best receiver \& the 1st nearest eavesdropper}
\begin{proposition} \label{P4}
{\rm The PNZ of the $k$-th best legitimate receiver in the presence of the 1st nearest non-colluding eavesdropper is given by }
\begin{equation} \label{Pspc_best_nearest}
\begin{split}
&\hspace{-2ex} \mathcal{P}_{nz,BN}=  1 - \frac{\epsilon_e }{\theta_e  \Gamma(k)} \\
& \hspace{-3ex}\times H_{3,2}^{1,3} \left[ { \frac{ 1}{\theta_e \varpi }\left( \frac{A_e}{A_{b1}} \right)^{\frac{1}{\delta}}\hspace{-0.5ex} \left| \hspace{-1ex}{\begin{array}{*{20}c}
   {(1,1),(1 -\mu_e,\frac{2}{\alpha_e}),(1-k,\frac{1}{\delta})}  \\
   {(1, \frac{1}{\delta}),(0,1)}   \\
\end{array}} \right.} \hspace{-1.5ex} \right]. 
\end{split}
\end{equation}
\end{proposition}
\begin{proof} 
See Appendix \ref{Appen_G}.
\end{proof}
\subsection{Performance Characterization of Ergodic Secrecy Capacity}
From the perspective of the eavesdroppers' received signal quality, the first nearest or best eavesdropper can achieve the highest composite channel gain. As such, the ergodic secrecy capacity can be similarly analyzed for the four considered scenarios. Motivated from (\ref{Cs_ergodic}), the ergodic secrecy capacity can be obtained from the difference of the ergodic capacities between the transmitter-legitimate receiver link and the transmitter-eavesdropper link \cite{7247765}. In accordance with our proposed user association method, i.e., the $k$-th nearest or best user, the ergodic capacity of the transmitter to the $k$-th nearest legitimate receiver, $ R_{N,k}^M$, and the transmitter to the $k$-th best legitimate receiver, $R_{B,k}^M$, are correspondingly obtained in the follow proposition in order to simplify our derivations of ergodic secrecy capacity.
\begin{proposition} \label{Propo_7}
The ergodic capacity of the transmitter to the $k$-th nearest or best legitimate user, $R_{N,k}^M$ and $R_{B,k}^M$, are respectively given by
\begin{subequations}
\begin{equation}
 R_{N,k}^M = \frac{\epsilon_k H_{3,3}^{2,3}\hspace{-0.5ex} \left[ { \frac{\eta_k A_k^{\frac{1}{\delta}}}{ \theta_k }\hspace{-0.5ex}  \left|\hspace{-1.2ex} {\begin{array}{*{20}c}
    {(1,1),(1,1),(1-\mu_k , \frac{2}{\alpha_k })}  \\
   {(1,1),(k,\frac{1}{\delta}),(0,1)}  \\
\end{array}} \right.} \hspace{-1ex}\right]}{\theta_k \Gamma(k)\ln 2} ,
\end{equation}
\begin{equation}
 R_{B,k}^M = \frac{\delta}{\Gamma(k) \ln 2} H_{3,2}^{2,2} \left[ { A_{b1}\eta_k^\delta  \left|\hspace{-1ex} {\begin{array}{*{20}c}
    {(1,\delta),(1,\delta))}  \\
   {(k,1),(1,\delta),(0,\delta)}  \\
\end{array}} \right.} \hspace{-1.5ex}\right],
\end{equation}
\end{subequations} 
where $H_{p,q}^{m,n}[.]$ is the Fox's $H$-function.
\end{proposition}
\begin{proof}
See Appendix. \ref{ErgSecrecyCapacity_Proof}.
\end{proof}
Similarly, the $R_{N,k}^W$ and $R_{B,k}^W$ can be easily derived by making some simple manipulations. Accordingly, considering either the $1$st nearest or best eavesdropper, by letting $k = 1$ for $R_{N,k}^W$ and $R_{B,k}^W$, and setting $R_{N,1}^W$ and $R_{B,1}^W$ as $R_{N}^W$ and $R_{B}^W$, then we have the following remark.
\begin{remark}
{\rm Taking account of the aforementioned four scenarios, the ergodic secrecy capacity are respectively given by}
\begin{itemize}
\item case 1: 
\begin{subequations}
\begin{equation} \label{ASC_NN}
\bar{C}_{s:k,NN} = [R_{N,k}^M - R_N^W]^+,
\end{equation}
\item case 2: 
\begin{equation} \label{ASC_BB}
\bar{C}_{s:k,BB} = [R_{B,k}^M - R_B^w]^+,
\end{equation}
\item case 3: 
\begin{equation} \label{ASC_NB}
\bar{C}_{s:k,NB} = [R_{N,k}^M - R_B^W]^+,
\end{equation}
\item case 4: 
\begin{equation} \label{ASC_BN}
\bar{C}_{s:k,BN} = [R_{B,k}^M - R_N^W]^+.
\end{equation}
\end{subequations}
\end{itemize}
\end{remark}
For the sake of showing brevity, the details are not given in (\ref{ASC_NN}-\ref{ASC_BN}), respectively, however, those can be easily obtained from Proposition. \ref{Propo_7} by making some simple algebraic substitutions.
\section{Numerical Results and Discussions} \label{SimulationResults}
For a given network configuration, the secrecy metrics, including the COP and PNZ, are under analysis in Section \ref{Sec_noncoll}. In this section, the accuracy of our analysis is validated by presenting numerical simulations. In the whole simulation configuration, it is assumed that $r = 10$ and the simulation solely takes places under $\alpha$-$\mu$ fading channels.

In addition, we will study the effects of the density, the path-loss exponent $\upsilon$, different $\alpha$-$\mu$ fading factors and dimensions of space on the secrecy metrics. Note that in our simulation, the WAFO toolbox of MATLAB \cite{Wafo2000} has been used to generate $\alpha$-$\mu$ variates.

It is very important to note that higher system performance is achieved at lower COP as well as higher PNZ probabilities.
\subsection{Results pertaining to COP}
This subsection studies the system performance with respect to the nearest and best legitimate receivers, and we provide a comparison between the two performances.

The $\mathcal{P}_{co,N}$ stated in \eqref{pco_nearest} versus the $k$-th nearest legitimate receiver under $\alpha$-$\mu$ fading is shown in Fig. \ref{Pco_index_case1}.
It demonstrates how the COP for the $k$-th nearest legitimate receiver is affected as the legitimate user's index increases, for various $\alpha$-$\mu$ fading scenarios. In addition, Fig. \ref{Pco_index_case1} also demonstrates the conformity of our analytical derivations to simulation outcomes.
\begin{figure}[!t]
\centering
\includegraphics[width=3.3in]{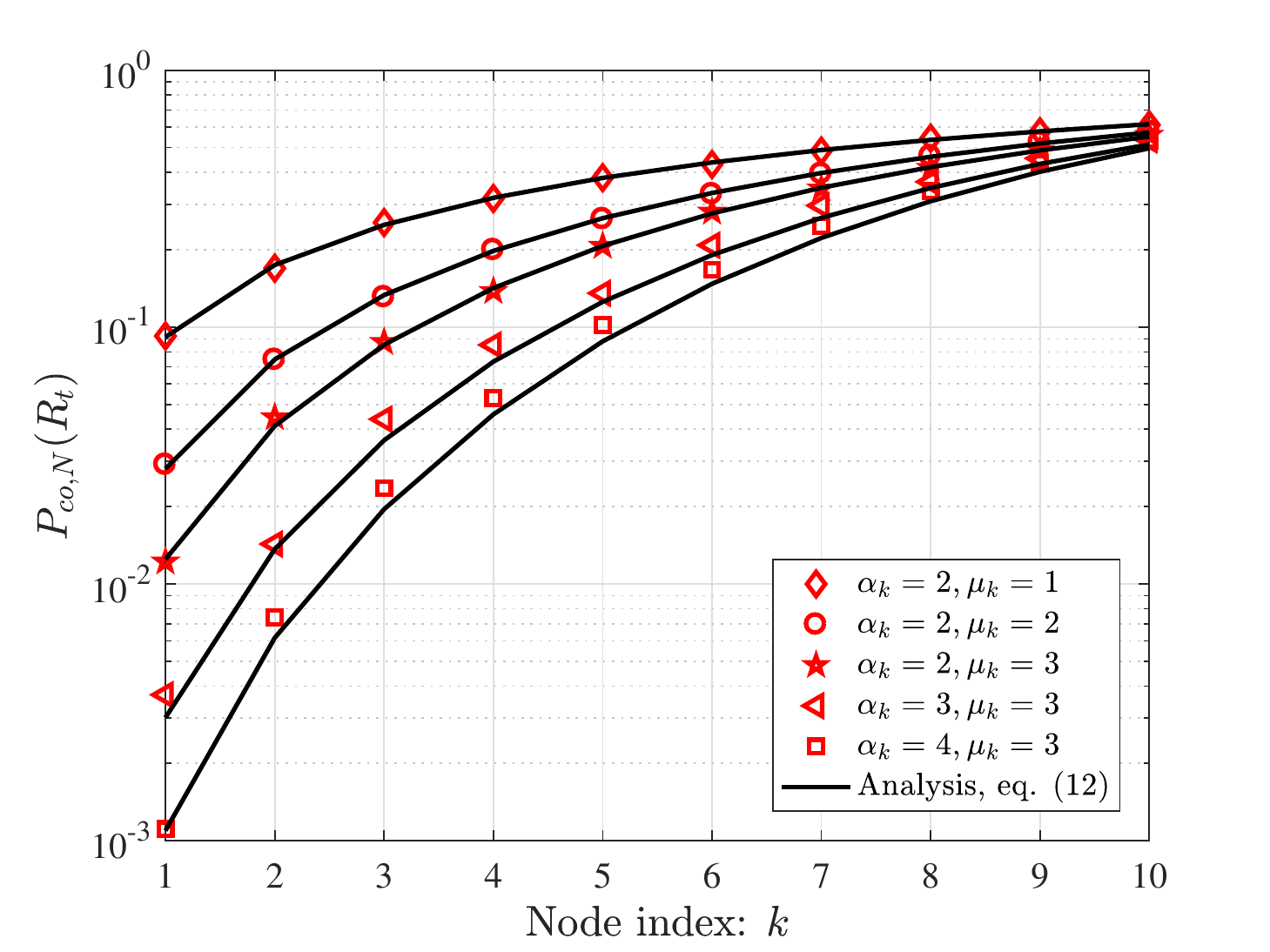}
\caption{\small $P_{co,N}$ versus the $k$-th nearest legitimate receiver for $\eta_k = 5$ dB, $\lambda_b = 1$, $N_a = N_b = 1$ $R_t =1$. }
\label{Pco_index_case1}
\end{figure}

The $\mathcal{P}_{co,B}$ drafted in \eqref{pco_best} versus $\lambda_b$ is illustrated and compared with the $\mathcal{P}_{co,B}$ in Fig. \ref{Pco_com_Best} for selected values of the $k$-th legitimate nearest/best receiver. From this graph, we obtain the conclusions that: (i) the connection outage occurs with a higher probability for larger index values and larger $\lambda_b$; (ii) since $\lambda_b$ grows in equal steps, the gap between $\mathcal{P}_{co,N}$ and $\mathcal{P}_{co,B}$ tends to be larger for higher index values $\lambda_b$. 
\begin{figure}[!t]
\centering
\includegraphics[width=3.3in]{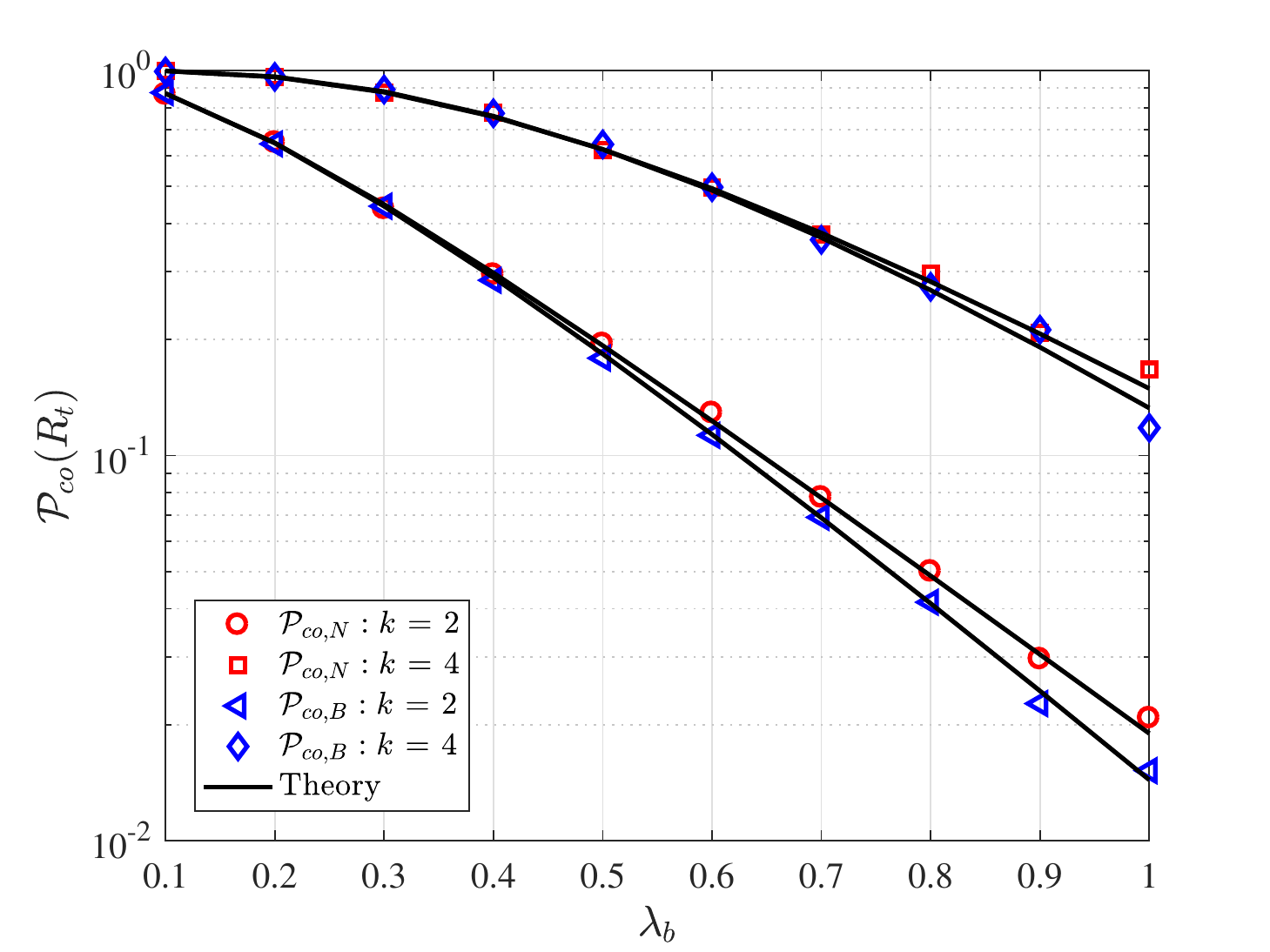}
\caption{\small $P_{co}$ versus $\lambda_b$ for selected $k$-th ($k\in \{ 2,4\}$) nearest/best user when $\eta_k = 0$ dB, $R_t =1$, $\alpha_k = 2$, $\mu_k = 3$, $\upsilon = 4$, $d = 2$. }
\label{Pco_com_Best}
\end{figure}

Having studied the performance with respect to the nearest and best legitimate receivers, we compare the COP for the 1st nearest and best legitimate receiver for various selected path-loss exponent $\upsilon$ values and $N_b$, in the next step. The result of this comparison is shown in Fig. \ref{Pco_com_Antenna}. Strikingly, one can conceive that on one hand, higher path-loss exponent always results in a higher probability of connection outage both for the $k$-th nearest and best receivers. On the other hand, the $k$-th best receiver always owns a relatively lower connection outage probability compared with the $k$-th nearest one, as predicted. In addition, the COP deserves with lower probability due to its better quality of received signal, as $N_b$ increase.
\begin{figure}[!t]
\centering
\includegraphics[width=3.3in]{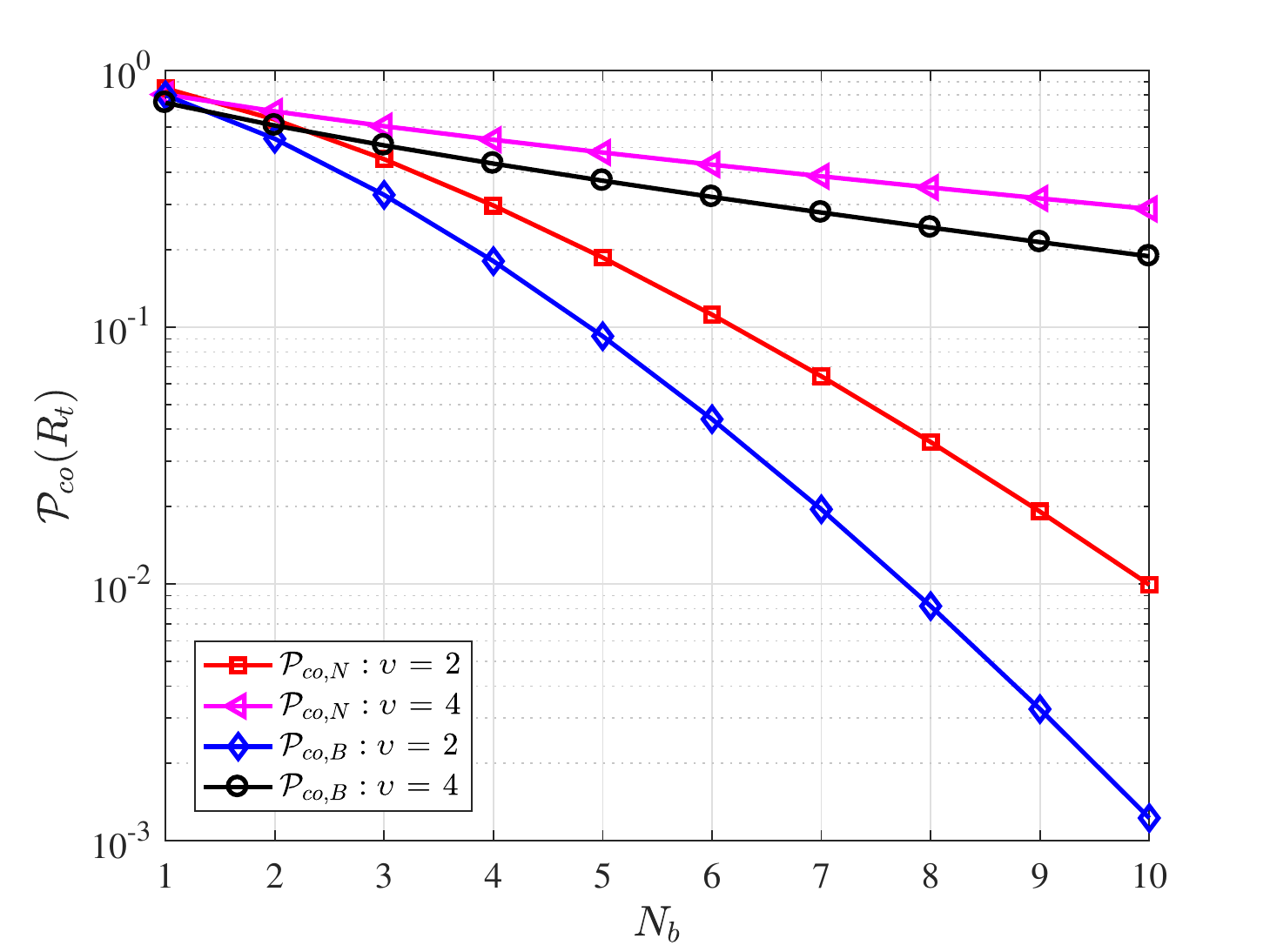}
\caption{\small Comparison of $\mathcal{P}_{co,N}$ to $\mathcal{P}_{co,B}$ versus $N_b$ for $\lambda_b = 0.1$, $\eta_k = -5$ dB, $\alpha_k=2$, $\mu_k =3$, $R_t =1$, $d=3$ and various path-loss exponent $\upsilon \in \{ 2,4\}$.}
\label{Pco_com_Antenna}
\end{figure} 
\subsection{Results pertaining to PNZ}
In this subsection, we study the probability of non-zero secrecy capacity in the presence of non-colluding eavesdroppers. Be reminded that higher PNZ probabilities indicate a better system performance. For the sake of simplicity, the first nearest/best eavesdropper is considered for evaluating the secrecy risk.
Figs. \ref{Pnz_NN_fading}--\ref{Fig_com_lambda} demonstrate the PNZ versus the $k$-th legitimate receiver in the presence of non-colluding eavesdroppers. It is easily observed that our theoretical analyses are in strong agreement with the simulation outcomes.

Fig. \ref{Pnz_NN_fading} plots the PNZ against the $k$-th nearest legitimate receiver's index for selected values of $\alpha$ and $\mu$ when the nearest eavesdropper is considered. It is observed here that almost for all values of the $k$-th user index, the PNZ performance is better (probability is higher) for smaller values of $\alpha$, $\mu_{\rm m}$ and $\mu_{\rm w}$.
\begin{figure}[!t]
\centering{\includegraphics[width=3.3in]{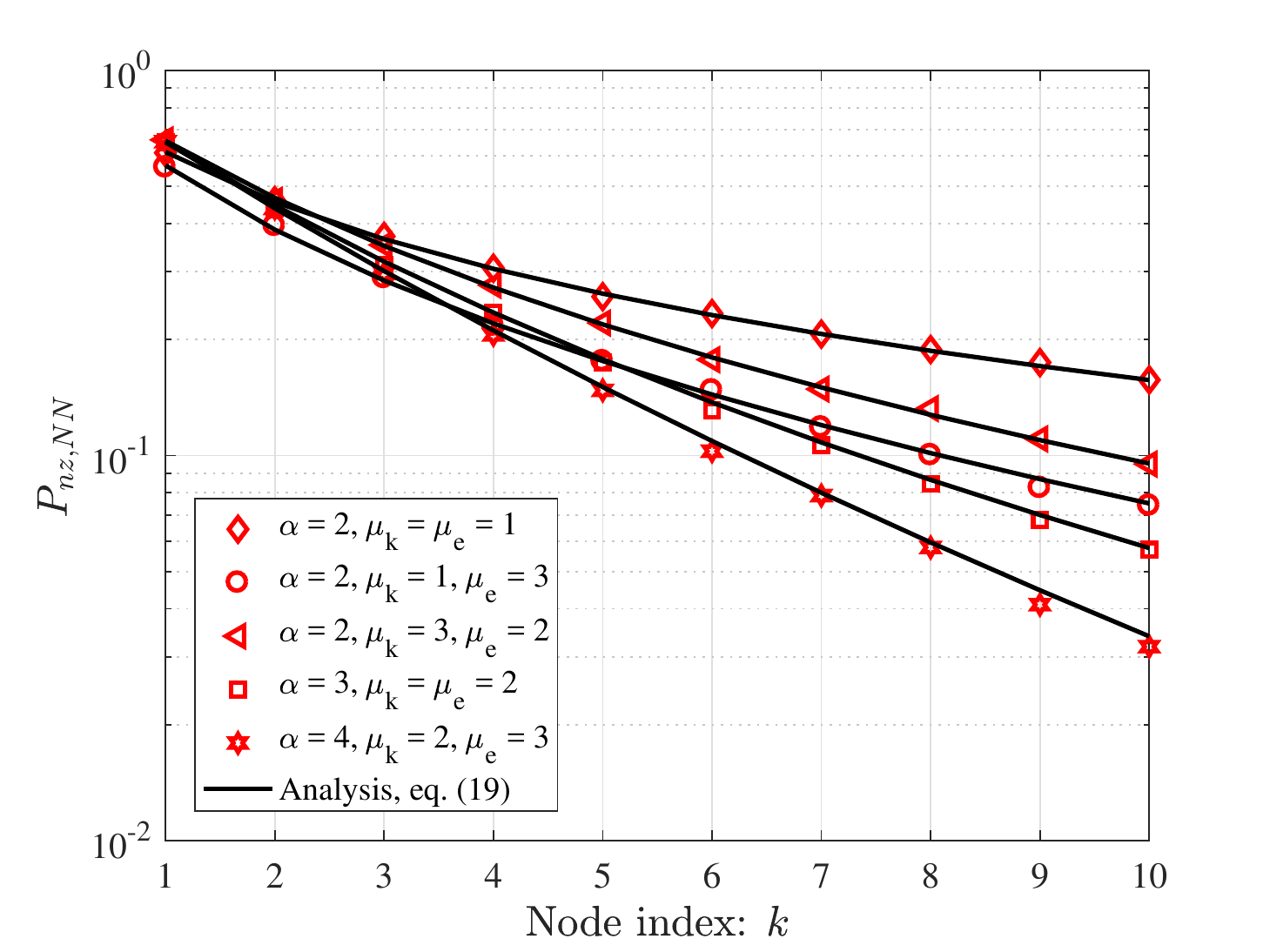}}
\caption{\small{$\mathcal{P}_{nz,NN}$ versus the $k$-th nearest legitimate receiver for $\varpi = 0$ dB, $N_a = N_b = N_e =1$, $\alpha_k = \alpha_e = \alpha$, $\lambda_b = 0.2$, $\lambda_e =0.1$, $d = 2$, $\upsilon = 2$.}}
\label{Pnz_NN_fading}
\end{figure}
\begin{figure}[!t]
\centering{\includegraphics[width=3.3in]{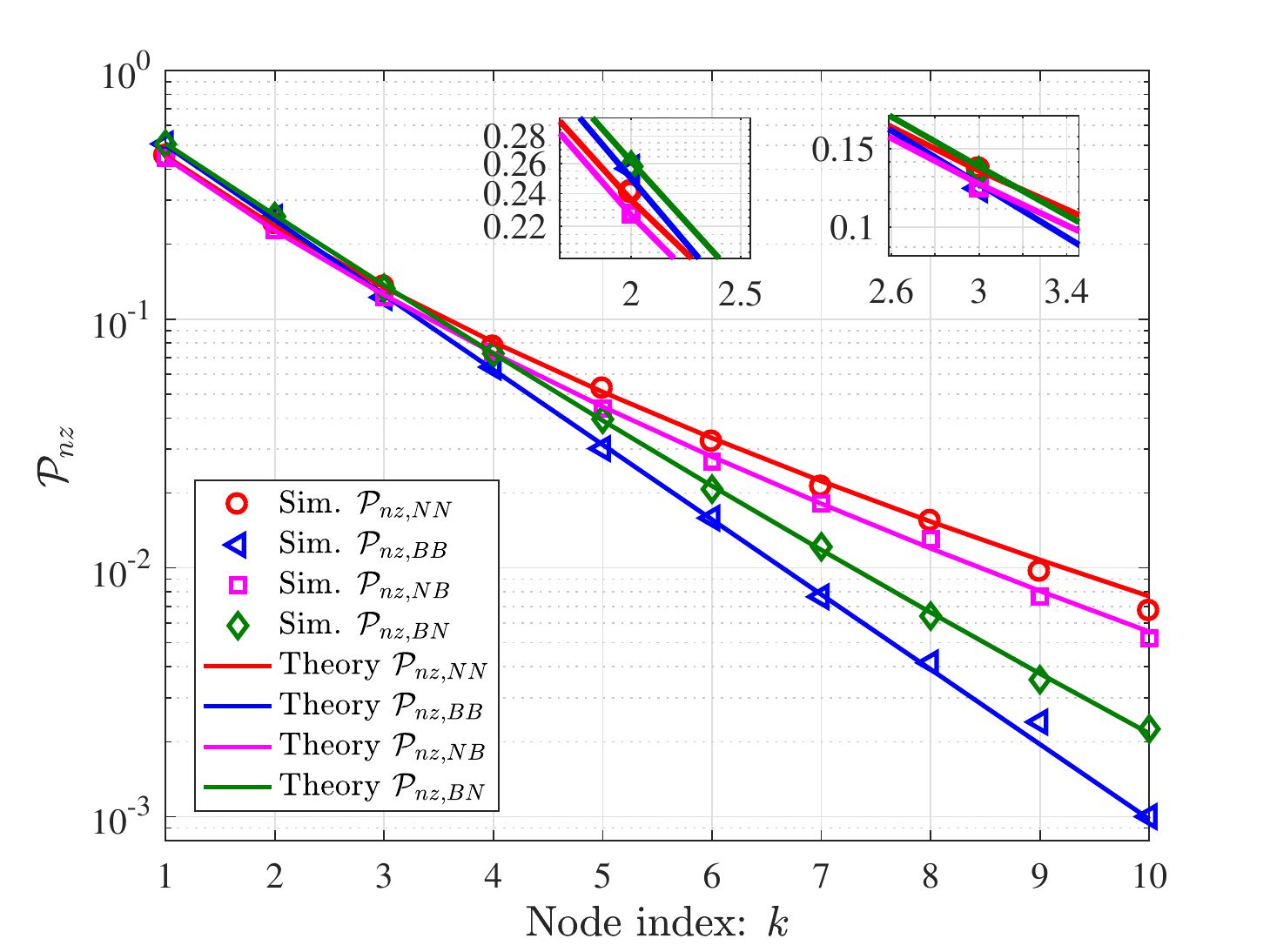}}
\caption{\small{$\mathcal{P}_{nz}$ versus the $k$-th legitimate receiver for $\varpi= 0$ dB, $\lambda_b = 0.2$, $\lambda_e =0.1$, $N_a =2, N_b = 1,N_e = 2$, $\alpha_k = 2$, $\mu_k = 1$, $\alpha_e = 2$, $\mu_e = 4$, $d= 2$, $\upsilon = 2$.}}
\label{Pnz_com_KthUser}
\end{figure}
\begin{figure}[!t]
\centering{\includegraphics[width=3.3in]{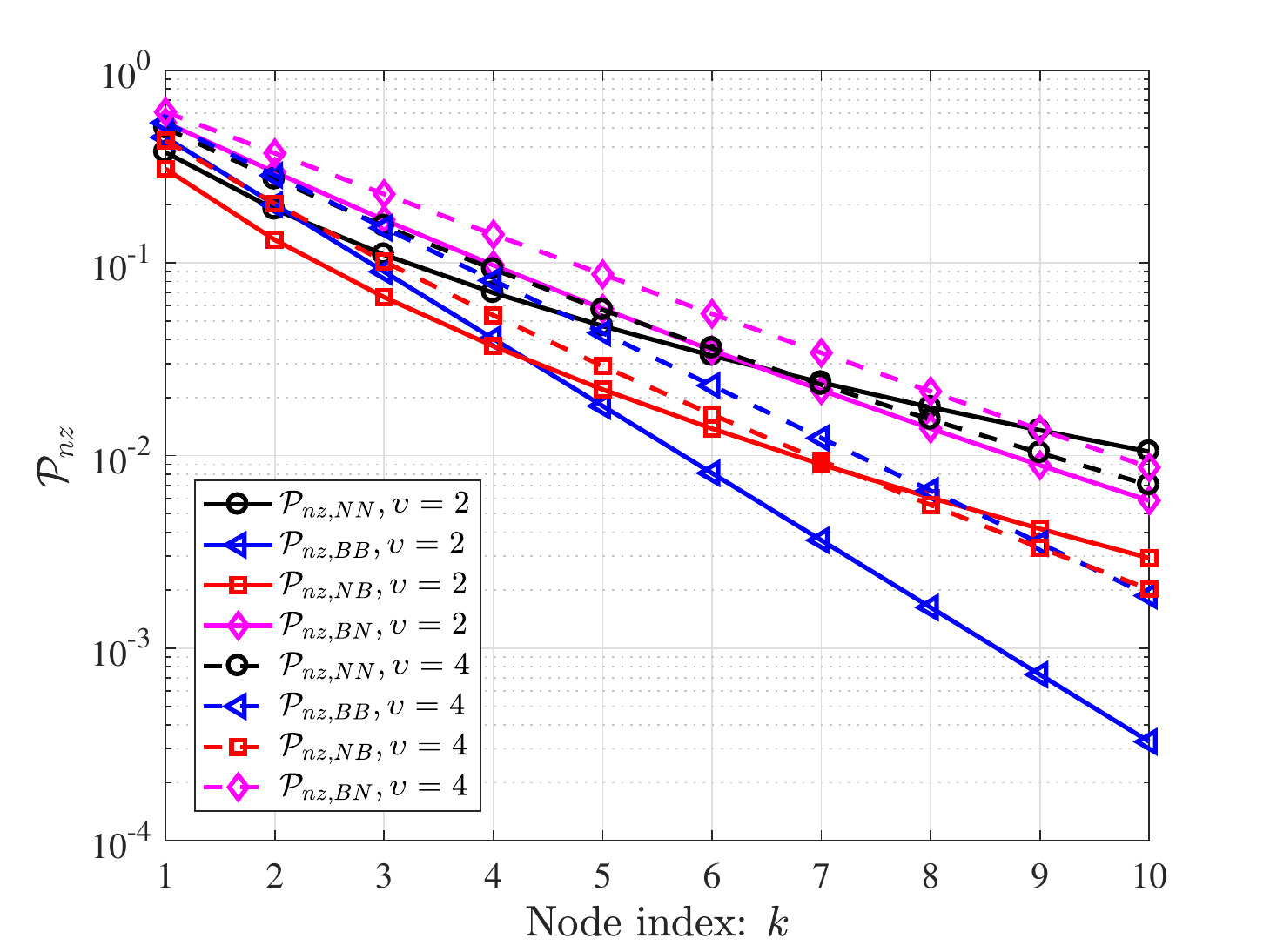}}
\caption{\small{$\mathcal{P}_{nz}$ versus the $k$-th nearest/best legitimate receiver for $\varpi = 0$ dB, $\lambda_b = 0.2$, $\lambda_e =0.1$, $N_a =2, N_b = 1,N_e = 2$, $\alpha_k = \alpha_e = \mu_k = 2$, $\mu_e = 3$, and $d = 3$.}}
\label{Pnz_com_PathLoss}
\end{figure}
\begin{figure}[!t]
\centering
\includegraphics[width=3.3in]{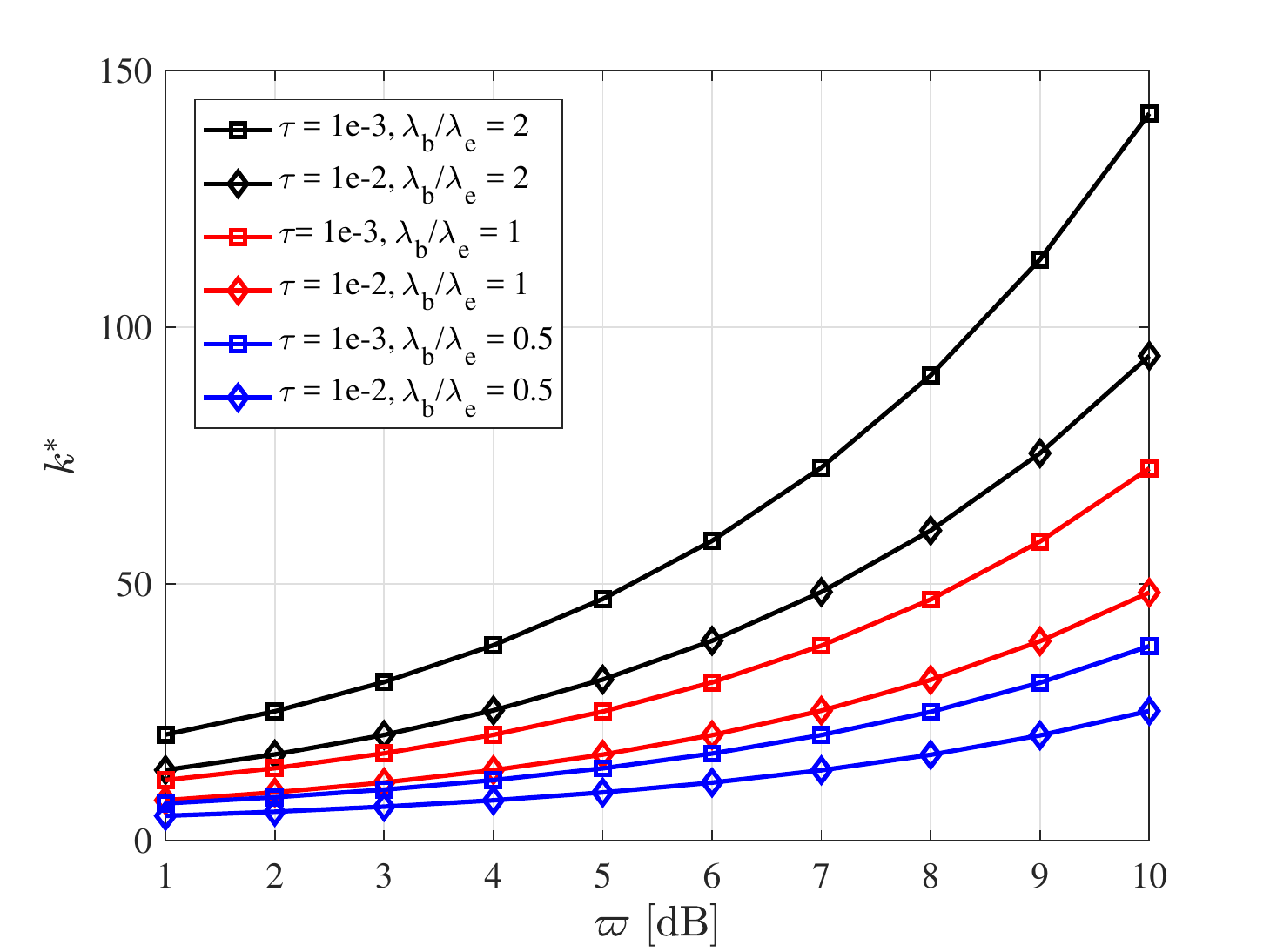}
\caption{\small The maximum size of the best ordered user $k^*$ versus $\varpi$ for selected values of $\tau$ and density ratios $\lambda_b/\lambda_e$, according to (\ref{EQ_bestPathGainEvek}), when $N_a = N_b = N_e = 1$, $\alpha_k= 3, \mu_k = 2$, $\alpha_k= 2, \mu_k = 3$, and $d= \upsilon = 2$.}
\label{Pnz_best_KthUSER}
\end{figure} 
\begin{figure}[!t]
\centering{\includegraphics[width=3.3in]{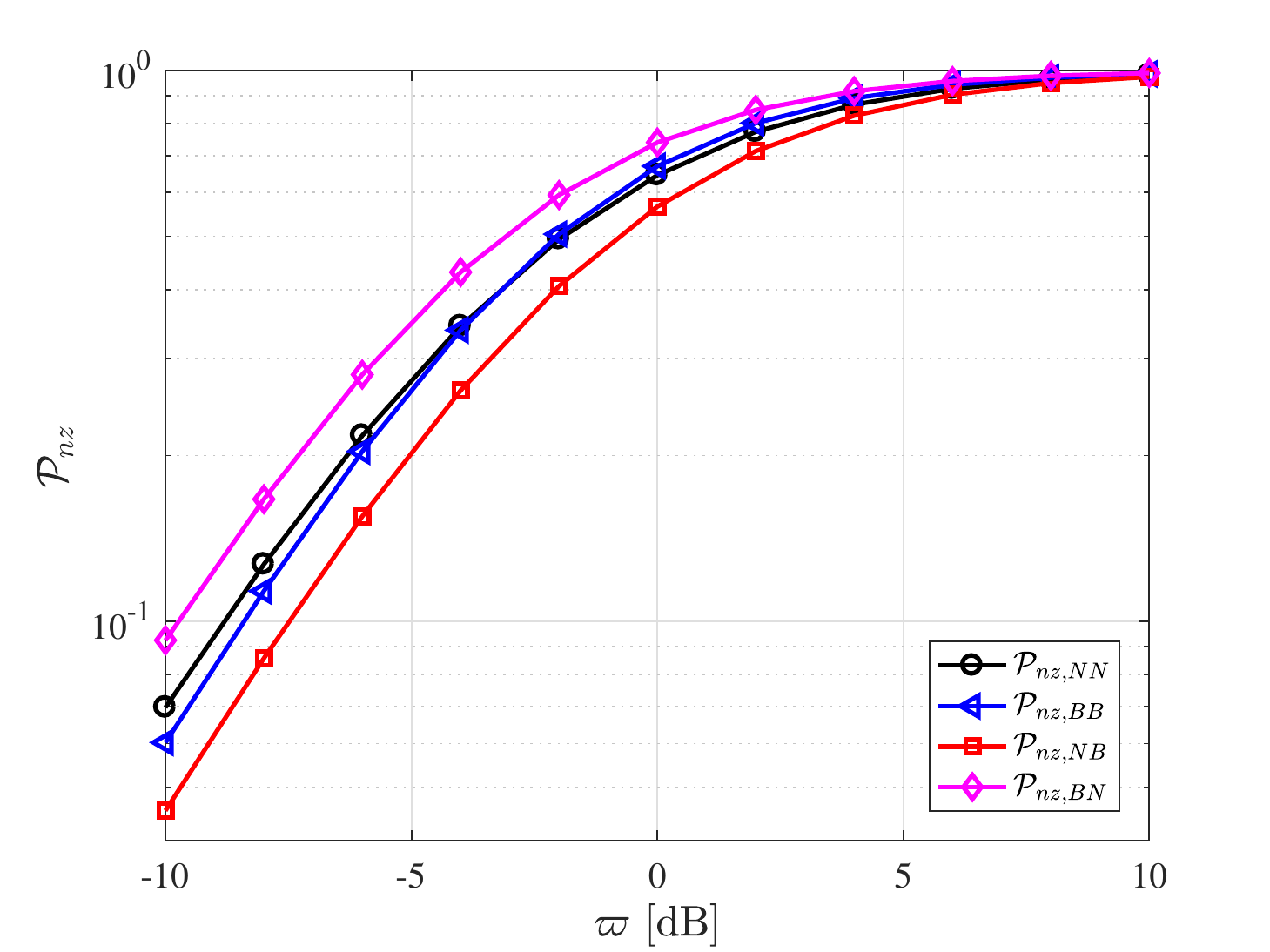}}
\caption{\small{$\mathcal{P}_{nz}$ versus $\varpi$ for the $1st$ nearest/best legitimate receiver for $\lambda_b = 0.2$, $\lambda_e =0.1$, $N_a = N_b = N_e = 2$, $\alpha_k = \alpha_e = 2$, $\mu_k = 2$, $\mu_e = 3$, $d = 3$ and $\upsilon = 2$.}}
\label{Pnz_com_EbNo}
\end{figure}

Fig. \ref{Pnz_com_KthUser} compares the PNZs given in (\ref{Pspc_nz1_nearest}), (\ref{PSPC_best_final}), (\ref{Pspc_nearest_best}) and (\ref{Pspc_best_nearest}) for the four scenarios, where the 1st nearest or best eavesdropper is considered. One can conceive that (i) our closed-form expressions are confirmed by the Monte-Carlo simulation outcomes; (ii) the $\mathcal{P}_{nz,BN}$ outperforms the other three scenarios when $k = 1, 2$, this trend is changing as $k$ reach 4, the probability of having a positive secrecy capacity drops in a descending order, namely, $\mathcal{P}_{nz,NN}> \mathcal{P}_{nz,NB}>\mathcal{P}_{nz,BN}>\mathcal{P}_{nz,BB}$. The reason behind lies in that two ordering key factors, i.e., distances and composite channel gain, are in turn playing a critical role on the secrecy performance especially as $k$ increases.
As shown in Fig. \ref{Pnz_com_PathLoss}, the influence of $\upsilon$ on the PNZ is demonstrated. As it can be readily observed, the PNZs tend to decrease as the $k$-th user index grows for all considered $\upsilon$.

Fig. \ref{Pnz_best_KthUSER} presents the maximum number of the $k$-th best users for a given probability constraint $\tau$. As illustrated in this figure, it can be easily seen that many more best users are permitted for higher $\varpi$ and higher $\lambda_b/\lambda_e$.  

As observed in Fig. \ref{Pnz_com_KthUser}, the $1$st nearest/best legitimate receiver is mostly endangered by the malicious eavesdropper. As a result, in the following three Figs, the impacts of $\varpi$, the receiving antenna numbers $N_b, N_e$, and the density of two kinds of receivers, $\lambda_b$ and $\lambda_e$ on the PNZ are investigated. In this case, the first nearest/best legitimate receiver is considered in Figs. \ref{Pnz_com_EbNo}-\ref{Fig_com_lambda}. 

In Fig. \ref{Pnz_com_EbNo}, the PNZs are anticipated to witness an increasing trend as $\varpi$ increases. It is intuitively observed that the 1st best user is guaranteed with a higher probability in the presence of the 1st nearest eavesdropper. Such a phenomenon repeats itself for the Figs \ref{Pnz_com_AnteNoBob} and \ref{Fig_com_lambda}.

To terminate the discussion, in Figs. \ref{Pnz_com_AnteNoBob} and \ref{Fig_com_lambda}, we present the PNZs against the number of receiving antennas and the densities, respectively.  
\begin{figure}[!t]
\centering{\includegraphics[width=3.3in]{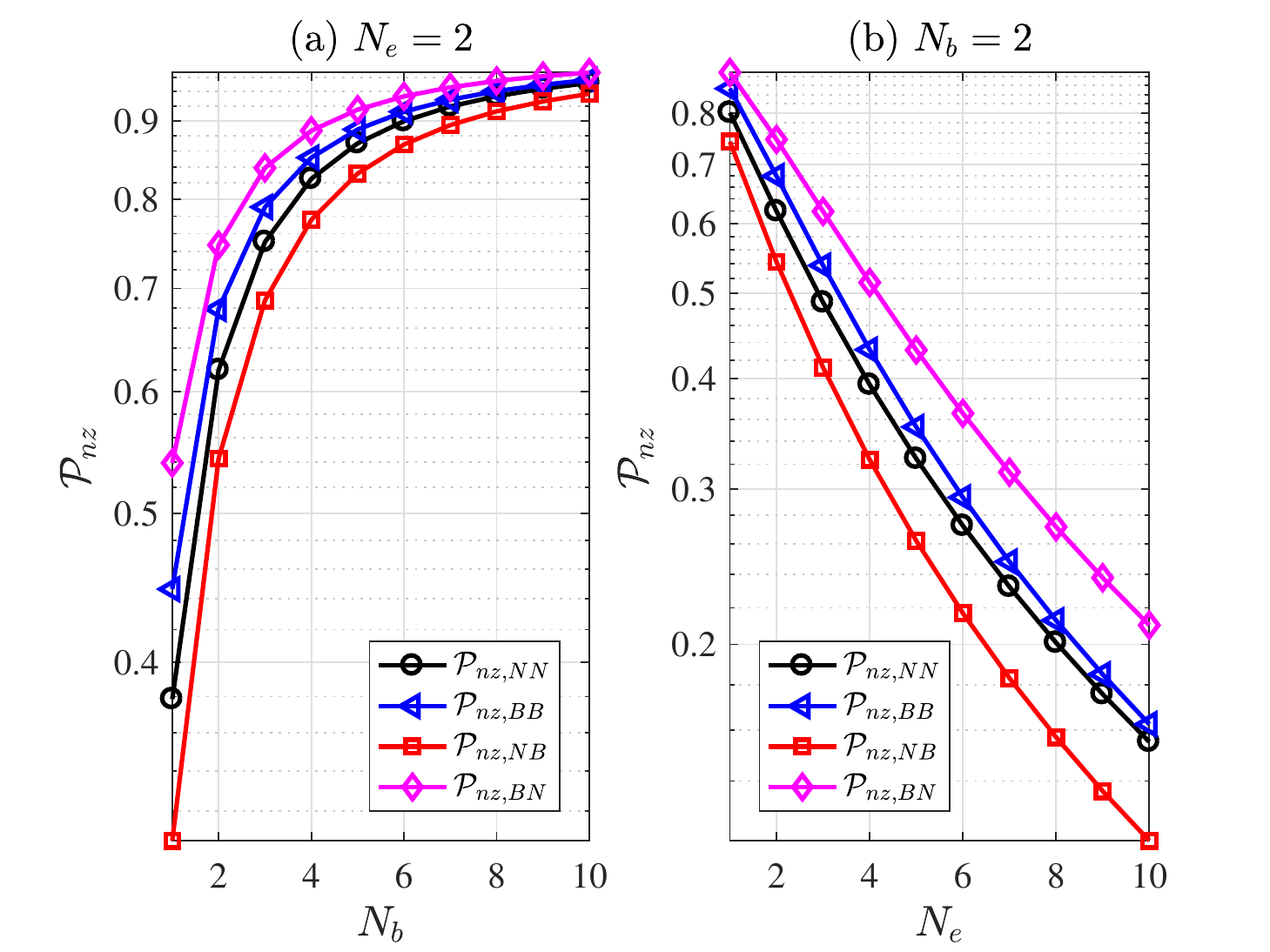}}
\caption{\small{$\mathcal{P}_{nz}$ versus the number of received antennas at the 1st nearest/best receivers for $\varpi = 10$ dB, $\lambda_b = 0.2$, $\lambda_e =0.1$, $\alpha_k = \alpha_e = 2$, $\mu_k = 1$, $\mu_e = 3$, $d = 3$, $N_a = 2$ and $\upsilon = 2$.}}
\label{Pnz_com_AnteNoBob}
\end{figure}
It is observed that an increased $N_b/N_e$ ratio indicates the legitimate receivers are much more capable to achieve a higher quality of receiving signals, which naturally yields a higher probability of positive secrecy capacity. It is validated by Fig. \ref{Pnz_com_AnteNoBob}(a). 

On the contrary, this trend is conversely preserved regardless of the $k$-th user index value. As $N_e/N_b$ increases, the $k$-th best legitimate receiver achieves the highest and second-highest probability of non-zero secrecy capacity (best performance), in the presence of the nearest/best eavesdropper, respectively, which are characterized by $\mathcal{P}_{nz,BN}$ and $\mathcal{P}_{nz,BB}$. Next, the 1st nearest legitimate receiver suffers more, resulting in a lower probability, as denoted by $\mathcal{P}_{nz,NN}$. Naturally, the worst performance is recoded when the system challenges against the best eavesdropper, described by $\mathcal{P}_{nz,NB}$.   

From the comparison of the PNZ against densities shown in Fig. \ref{Fig_com_lambda}, one can conclude that (i) conditioned on a given $\lambda_b$, the higher $\lambda_e$ indicates a system with relatively more eavesdroppers. An increase in the number of colluding eavesdroppers progressively endangers the legitimate link, i.e., probability becomes worse (lower) for higher number of eavesdroppers; (ii) for a fixed number of eavesdroppers, lower $\lambda_b$ values result in worse performance, i.e., lower probability of non-zero secrecy capacity.
\begin{figure}[!t]
\centering{\includegraphics[width=3.3in]{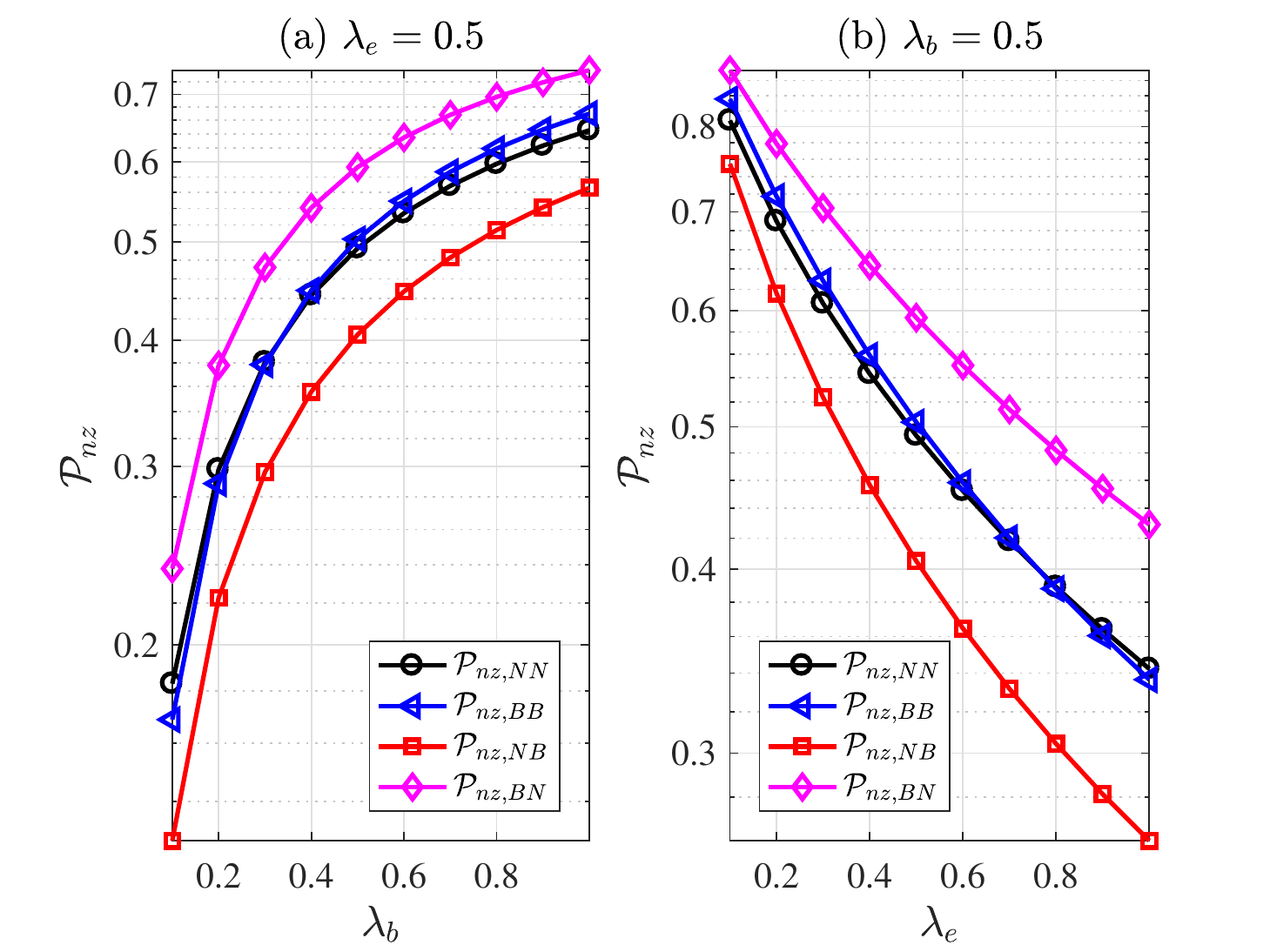}}
\caption{\small{$\mathcal{P}_{nz}$ versus the density of 1st nearest/best receivers for $\varpi = 10$ dB, $N_a = N_b = N_e = 2$, $\alpha_k = \alpha_e = 2$, $\mu_k = 2$, $\mu_e = 3$, $d = 3$ and $\upsilon = 2$.}}
\label{Fig_com_lambda}
\end{figure}
\subsection{Results pertaining to ergodic secrecy capacity}
Fig. \ref{ASC_com_kth} plots the ergodic secrecy capacity versus the $k$-th nearest or best legitimate receiver, while in the presence of the 1st nearest or best eavesdropper, respectively. Again, the same conclusion can be obtained: the ergodic secrecy capacity, as depicted in case 4, outperforms the other 3 cases. 
\begin{figure}[!t]
\centering{\includegraphics[width=3.3in]{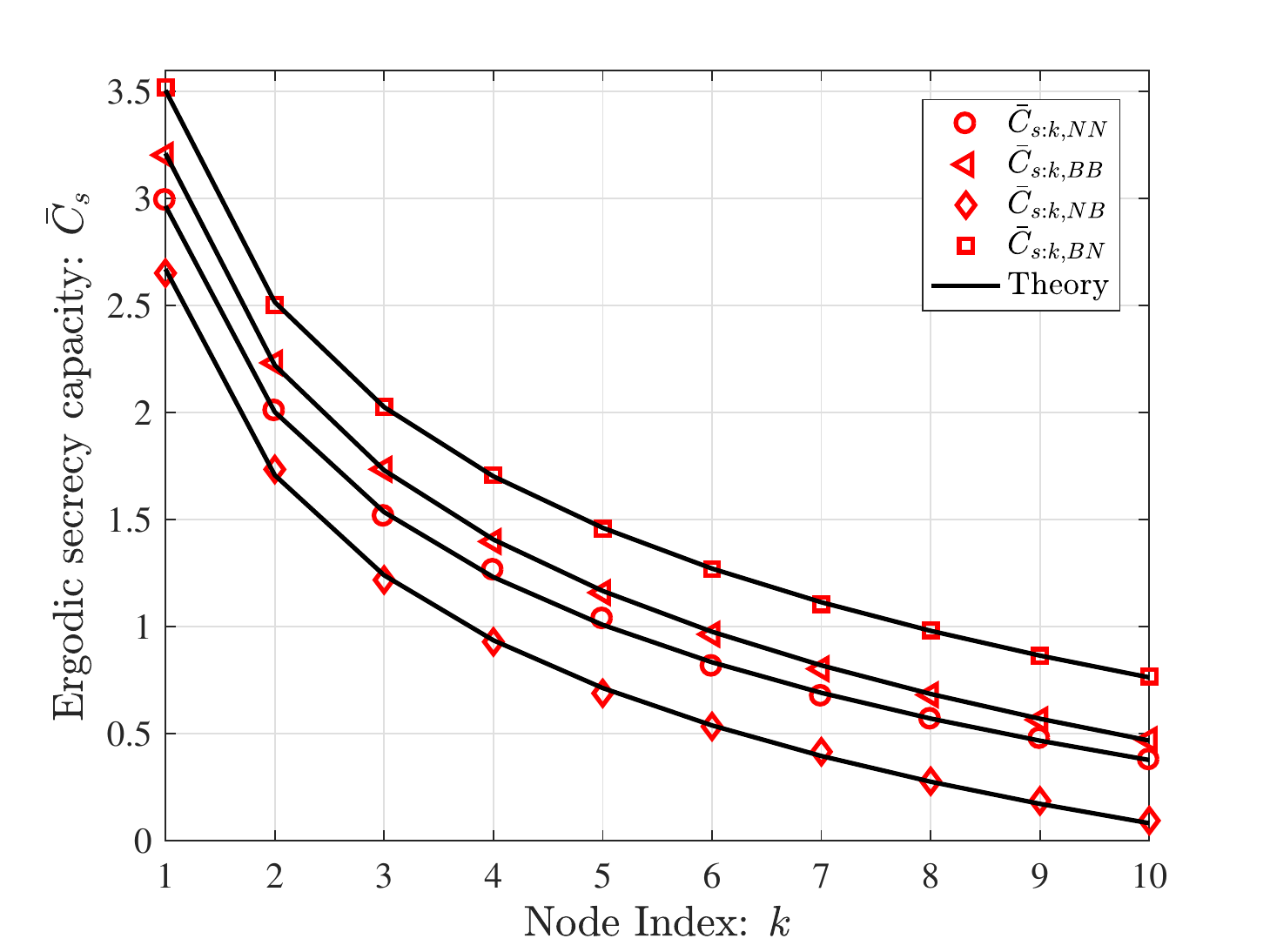}}
\caption{\small{$\bar{C}_{s}$ versus the $k$-th nearest/best legitimate receiver for $\lambda_b = \lambda_e = 1$, $N_a = N_b = N_e = 1$, $\alpha_k = \alpha_e = 2$, $\mu_k = \mu_e = 1$, $d = 2$ and $\upsilon = 2$, $\eta_k = 15$ dB, $\eta_e = 0$ dB.}}
\label{ASC_com_kth}
\end{figure}
\section{Conclusion} \label{Concluding}
In the context of this paper, we investigated the secrecy performance of HPPP-based random MIMO wireless networks over $\alpha$-$\mu$ fading channels for the first time. For the purpose of evaluating the secrecy performance of such a network, the COP, PNZ and ergodic secrecy capacity for the $k$-th nearest/best legitimate receiver in the presence of non-colluding eavesdroppers are derived and quantified with closed-form expressions. The accuracy of our analytical derivations are further successfully confirmed by simulation outcomes. Remarkable observations are drawn from the numerical results obtained in this paper. 

Indeed, the secrecy performance metrics are influenced by the density of users, the path-loss exponent, the number of transmitting and receiving antennas, as well as the fading parameters. In addition, the secrecy performance regarding the $k$-th best legitimate receiver outperforms that of the $k$-th nearest one. Hence, the nearest path does not necessarily provide the best secrecy performance. This paper's results and outcomes regarding parameters that influence secrecy performance will enable researchers or wireless system designers to quickly evaluate system performance and determine the optimal available parameter choices when facing different security risks. Finally, inspired from \cite{7143332}, future works will focus on using the beamforming deploying artificial noise technique over the homogeneous stochastic MIMO wireless network.
\appendices
\section{Derivation of $f_{\frac{g_k}{r_k^{\upsilon}}}(z)$} \label{Appen_A}
Setting $Z=\frac{g_k}{r_k^{\upsilon}}$, the PDF of $Z$ can be assessed by the ratio of $g_k$ and $r_k^{\upsilon}$, given by the following form
\begin{equation} \label{Appen_A_step1}
\begin{split}
f_{\frac{g_k}{r_k^{\upsilon}}}(z)& = \int_0^{\infty}yf_{g_k}(yz)f_{r_k^\upsilon}(y)dy\\
& \mathop=^{(b)}  \frac{{ \delta A_k^k \epsilon_k}}{\Gamma(k)} \int_0^\infty  y^{k\delta} \exp(-A_k y^{\delta}) \\
& \quad\times H_{0,1}^{1,0} \left[ { \theta_k z y \left| \hspace{-0.5ex}{\begin{array}{*{20}c}
    -  \\
   {(\mu_k - \frac{2}{\alpha _k }, \frac{2}{\alpha _k })}  \\
\end{array}} \right.} \hspace{-1ex}\right] dy,
\end{split}
\end{equation}
where $f_{r_k^\upsilon}(y)=\exp(-A_k y^{\delta})\frac{\delta(A_k y^\delta)^k}{y\Gamma(k)}$,  $A_k = \pi \lambda_b$ \cite[eq. (5)]{7136149}, $(b)$ is developed by substituting (\ref{PDFgi}).

Since the exponential function can be expressed in terms of Fox's $H$-function \cite[eq. (17)]{6754116}, given as 
\begin{equation} \label{Appen_A_step2}
\exp(  - A_k y^{\delta} )=\frac{1}{\delta } H_{0,1}^{1,0} \left[ { A_k^{\frac{1}{\delta}}y \left| {\begin{array}{*{20}c}
    -  \\
   {(0, \frac{1}{\delta})}  \\
\end{array}} \right.} \right],
\end{equation}
subsequently, substituting (\ref{Appen_A_step2}) into (\ref{Appen_A_step1}) and using the Mellin transform of the product of two Fox's $H$-function \cite[eq. (2.25.1.1)]{prudnikov1990integrals}, the proof is eventually concluded.
\vspace{-0.3cm}
\section{Derivation of $F_{\frac{g_k}{r_k^{\upsilon}}}(z)$} \label{Appen_B}
Essentially, $F_{\frac{g_k}{r_k^{\upsilon}}}(z)$ can be mathematically expressed as 
\begin{equation} \label{Pcoexpression}
\begin{split}
&F_{\frac{g_k}{r_k^{\upsilon}}}(z) = \int_0^{\infty}F_{g_k}(yz)f_{r_k^\upsilon}(y)dy \\
& =1 -  \frac{ \delta\epsilon_k A_k^k}{ \theta_k \Gamma(k)} \int_0^{\infty}y^{k\delta-1} \exp(-A_k y^\delta) \\
& \hspace{3ex}\times H_{1,2}^{2,0} \left[ { A_k^{\frac{1}{\delta}}y \left| {\begin{array}{*{20}c}
    {(1,1)}  \\
   {(0,1),(\mu_k, \frac{2}{\alpha _k })}  \\
\end{array}} \right.} \right]   dy ,
\end{split}
\end{equation}
by using (\ref{Appen_A_step2}) and with the aid of \cite[eq. (2.25.1.1)]{prudnikov1990integrals}, the proof is finally achieved.
\vspace{-0.3cm}
\section{Proof of Lemma 2} \label{Appen_C}
The intensity function of $\Psi = \{ r_k ^\upsilon\}$ can be derived from $\mathbb{E}\{ \Psi[0,x) \}=\lambda_b c_dx^\delta$ by utilizing the mapping theorem \cite[Corollary 2.a]{4675736}, i.e., 
$\lambda_\Psi = \lambda_b c_d \delta x^{\delta-1}$.

The intensity of $\Xi_k$ is obtained by applying the displacement theorem \cite{haenggi2012stochastic} as follows 
\begin{equation} \label{lambda_xi}
\begin{split}
&\lambda_{\Xi_k} = \int_0^\infty \lambda_\Psi \rho(x,y) dx  = \int_0^\infty \lambda_\Psi \frac{x}{y^2}f_{g_k}(x/y) dx \\
&= \int_0^\infty  \lambda_b c_d \delta\frac{x^\delta}{y^2}f_{g_k}(x/y) dx \\
&\mathop=^{(c)}  \lambda_b c_d \delta y^{\delta-1} \underbrace{ \int_0^\infty z^\delta f_{g_k}(z) dz}_{U_4}, 
\end{split}
\end{equation}
where $(c)$ is obtained by changing the variable $z = x/y$. The integral in (\ref{lambda_xi}) is then solved as
\begin{equation} \label{I_1}
\begin{split}
U_4& = \int_0^\infty \frac{{\alpha _k z^{\frac{{\alpha _k \mu _k }}{2} + \delta - 1} }}{{2\Omega _k^{\frac{{\alpha _k \mu _k }}{2}} \Gamma \left( {\mu _k } \right)}}\exp \Bigg( { - \left( {\frac{z}{{\Omega _k }}} \right)^{\frac{{\alpha _k }}{2}} } \Bigg)dz \\
&\mathop=^{(d)} \frac{\Gamma(\mu_k+\frac{2\delta}{\alpha_k})\Omega_k^\delta}{\Gamma(\mu_k)},
\end{split}
\end{equation}
where $(d)$ holds by using \cite[eq. (3.381.10)]{gradshteyn2014table}. The proof is eventually concluded by substituting (\ref{I_1}) into (\ref{lambda_xi}).
\section{Proof of Lemma 3} \label{Appen_D}
By using \cite[Lemma 2]{7742344}, we have 
\begin{equation}  \label{CDF_best_xi}
\begin{split}
&F_{\xi_k}(x) = \mathcal{P}r(\xi_k <x) = 1 - \mathcal{P}r(\Xi[0,x]<k)\\
& = 1 - \sum_{n=0}^{k-1}\exp\left( -\int_0^x \lambda_{\Xi_k}(y)dy\right)\frac{\left(\int_0^x \lambda_{\Xi_k}(y)dy\right)}{n!}\\
&= 1- \sum_{n=0}^{k-1}\exp(-A_{b1}x^\delta)\dfrac{(A_{b1}x^\delta)^n}{n!}  = \frac{\gamma\left( k, A_{b1} x^\delta\right)}{\Gamma(k)}.
\end{split}
\end{equation}
When taking the derivative of (\ref{CDF_best_xi}), all terms in the sum are canceled out but the one for $n-1$. The PDF of $\xi_k$ becomes 
\begin{equation} \label{PDF_best_xi}
f_{\xi_k}(x) = \exp\left(-A_{b1}x^{\delta}\right)\frac{\delta(A_{b1}x^{\delta})^k}{x\Gamma(k)}.
\end{equation}
Therefore, the composite channel gain for the $k$-th best user can be termed as 
\begin{equation}
\begin{split}
F_{\frac{1}{\xi_k}}&=\mathcal{P}r\left(\frac{1}{\xi_k} < z\right)= 1 - F_{\xi_k}\left(\frac{1}{z} \right) \\
&= 1 -  \frac{\gamma\left( k, A_{b1} z^{-\delta}\right)}{\Gamma(k)}  = \frac{\Gamma\left( k, A_{b1} z^{-\delta}\right)}{\Gamma(k)}.
\end{split}
\end{equation}
Herein, the last step is derived from \cite[eq. (8.356.3)]{gradshteyn2014table}. By taking the derivative of $F_{\frac{1}{\xi_k}}(z)$ in terms of $z$, the PDF of $\frac{1}{\xi_k}$ is directly obtained. 
\section{Derivation of $\mathcal{P}_{nz,NN}$ in (\ref{Pspc_nz1_nearest})} \label{Appen_E}
Inspired by Lemma \ref{Lemma1}, $\mathcal{P}_{nz,NN}$ can be essentially derived as follows
\begin{equation}
\begin{split}
&\hspace{-2ex}\mathcal{P}_{nz,NN} =  \int_0^\infty F_{\frac{ g_e}{r_e^\upsilon}}(\varpi y) f_{\frac{ g_k}{r_k^\upsilon}} \left( y \right) dy \\
& = 1 - \frac{\epsilon_k \epsilon_e} {\theta_e A_k^{\frac{1}{\delta}} \Gamma(k) } \int_0^\infty H_{1,1}^{1,1} \left[ { \frac{\theta_k}{ A_k^{\frac{1}{\delta}} } y \left| {\begin{array}{*{20}c}
    {(1 - k - \frac{1}{\delta}, \frac{1}{\delta})}  \\
   {(\mu_k -\frac{2}{\alpha_k} , \frac{2}{\alpha_k})}  \\
\end{array}} \right.} \hspace{-1.5ex} \right] \\
& \hspace{3ex} \times H_{2,2}^{2,1} \left[ { \frac{\theta_e}{A_e^{\frac{1}{\delta}}} \varpi y \left| {\begin{array}{*{20}c}
    {(0, \frac{1}{\delta}),(1,1)}  \\
   {(0,1),(\mu_e , \frac{2}{\alpha_e})}  \\
\end{array}} \right.} \right], 
\end{split}
\end{equation}  
with the help of \cite[eq. (2.25.1.1)]{prudnikov1990integrals}, the proof is accomplished.
\section{Derivation of $\mathcal{P}_{nz,NB}$ in (\ref{Pspc_nearest_best})} \label{Appen_F}
Thanks to the CDF of $\frac{ g_k}{r_k^\upsilon}$ and PDF of $\xi_e$, respectively given in (\ref{CDF_nearest}) and (\ref{PDF_best_xi}), the expression $\mathcal{P}_{nz,NB}$ can be easily stated as 
\begin{equation} \label{Step1_AppF}
\begin{split}
&\mathcal{P}_{nz,NB} = 1 - \int_0^\infty F_{\frac{ g_k}{r_k^\upsilon}} \left( \frac{1}{\varpi y} \right) f_{\xi_e}(y)dy \\
& = \frac{2 \delta A_{e1}}{\alpha_k \Gamma(\mu_k) \Gamma(k)} \int_0^\infty y^{\delta -1} \exp(-A_{e1}y^{\delta}) \\
&\hspace{3ex} \times H_{2,2}^{2,1} \left[ { \frac{\varpi_k}{\varpi A_k^{\frac{1}{\delta}} y}   \left| {\begin{array}{*{20}c}
   {(1-k,\frac{1}{\delta}),(1,1)}  \\
   {(0,1),(\mu_k, \frac{2}{\alpha_k})}   \\
\end{array}} \right.} \right] dy .
\end{split}
\end{equation}
By using (\ref{Appen_A_step2}) and with the assistance of the property of Fox's $H$-function \cite[eq. (8.3.2.7)]{prudnikov1990integrals},
\begin{equation} \label{Step2_AppF}
\begin{split}
 &H_{2,2}^{2,1} \left[ { \frac{\varpi_k}{\varpi A_k^{\frac{1}{\delta}} y}   \left| {\begin{array}{*{20}c}
   {(1-k,\frac{1}{\delta}),(1,1)}  \\
   {(0,1),(\mu_k, \frac{2}{\alpha_k})}   \\
\end{array}} \right.} \right] \\
&\hspace{5ex} = H_{2,2}^{1,2} \left[ { \frac{\varpi A_k^{\frac{1}{\delta}} y} {\varpi_k}  \left| {\begin{array}{*{20}c}  
   {(1,1),(1 - \mu_k, \frac{2}{\alpha_k})}   \\
   {(k,\frac{1}{\delta}),(0,1)}  \\
\end{array}} \right.} \right].
\end{split}
\end{equation}
 $\mathcal{P}_{nz,NB}$ can be further developed as 
\begin{equation} \label{Step3_AppF}
\begin{split}
&\hspace{-2ex}\mathcal{P}_{nz,NB} =\frac{2  A_{e1}}{\alpha_k \Gamma(\mu_k) \Gamma(k)} \int_0^\infty y^{\delta -1} H_{0,1}^{1,0} \left[ { A_{e1}^{\frac{1}{\delta}} y \left|\hspace{-1ex} {\begin{array}{*{20}c}
    {-}  \\
   {(0,\frac{1}{\delta})}  \\
\end{array}} \right.}\hspace{-1ex} \right]  \\
& \hspace{3ex}\times H_{2,2}^{1,2} \left[ { \varpi \Omega_k A_k^{\frac{1}{\delta}} y \left| {\begin{array}{*{20}c}
   {(1-k,\frac{1}{\delta}),(1,\frac{2}{\alpha_k})}  \\
   {(\mu_k, \frac{2}{\alpha_k}),(0,\frac{2}{\alpha_k})}   \\
\end{array}} \right.} \right] dy, 
\end{split}
\end{equation}
afterwards, performing the Mellin transform of the product of two Fox's $H$-functions \cite[eq. (2.25.1.1)]{prudnikov1990integrals}, the proof is eventually obtained.
\section{Derivation of $\mathcal{P}_{nz,BN}$ in (\ref{Pspc_best_nearest})} \label{Appen_G}
The $\mathcal{P}_{nz,BN}$ in (\ref{Pnz_BN}) can be tracked from the PDF of $\xi_k$ and the CDF of $\frac{g_e}{r_e^\upsilon}$, $\mathcal{P}_{nz,4}$ is thus given by
\begin{equation} \label{Step1_AppG}
\begin{split}
& \mathcal{P}_{nz,BN} = \int_0^\infty F_{\frac{g_e}{r_e^{\upsilon}}}\left( \frac{\varpi}{y}\right)f_{\xi_k}(y) dy\\
& = 1 - \frac{\epsilon_e \delta A_{b1}^{\frac{1}{\delta}} }{\theta_e \Gamma(k)}  \int_0^\infty y^{k\delta -1} \exp(-A_{b1} y^\delta) \\
& \hspace{3ex}\times H_{2,2}^{2,1} \left[ {   \frac{\theta_e \varpi }{A_e^{\frac{1}{\delta}} y } \left| {\begin{array}{*{20}c}
   {(0,\frac{1}{\delta}),(1,1)}  \\
   {(0,1),(\mu_e, \frac{2}{\alpha_e})}   \\
\end{array}} \right.} \right] dy.
\end{split}
\end{equation}
Subsequently, following the similar steps as (\ref{Step1_AppF}-\ref{Step3_AppF}), the proof is easily proved.
\section{Derivation of Proposition (\ref{Propo_7})} \label{ErgSecrecyCapacity_Proof}
As the very beginning, the logarithm function and exponential function can be alternatively rewritten in terms of the Fox's $H$-function \cite[eq. (1.7.2)]{mathai1978h} and \cite[eq. (8.4.6.5)]{prudnikov1990integrals}
\begin{equation}
\log_2(1+x) = \frac{1}{\ln 2} H_{2,2}^{1,2} \left[ { x \left|\hspace{-1ex} {\begin{array}{*{20}c}
    {(1,1),(1,1)}  \\
   {(1,1),(0,1)}  \\
\end{array}} \right.}\hspace{-1ex} \right],
\end{equation}
\begin{equation}
\exp(-x) =  H_{0,1}^{1,0} \left[ { x \left|\hspace{-1ex} {\begin{array}{*{20}c}
    {-}  \\
   {(0,1)}  \\
\end{array}} \right.}\hspace{-1ex} \right].
\end{equation}

\begin{equation}
\begin{split}
R_{N,k}^M & = \mathbb{E}_{\frac{ g_k}{r_k^{\upsilon}}}\left[ \log_2\left(1 + \frac{\eta_k g_k}{r_k^{\upsilon}} \right)\right] \\
& =\frac{\epsilon_k}{A_k^{\frac{1}{\delta}} \Gamma(k) \ln 2 } \int_0^\infty H_{2,2}^{1,2} \left[ { \eta_k y \left|\hspace{-1ex} {\begin{array}{*{20}c}
    {(1,1),(1,1)}  \\
   {(1,1),(0,1)}  \\
\end{array}} \right.}\hspace{-1ex} \right] \\
&\quad\quad\times H_{1,1}^{1,1} \left[ { \frac{\theta_k y}{ A_k^{\frac{1}{\delta}}}  \left| {\begin{array}{*{20}c}
    {(1 - k - \frac{1}{\delta},\frac{1}{\delta})}  \\
   {(\mu_k - \frac{2}{\alpha_k }, \frac{2}{\alpha_k })}  \\
\end{array}} \right.} \right] dy.
\end{split}
\end{equation}
Next, applying the Mellin transform of the product of two Fox's $H$-function \cite[eq. (2.25.1.1)]{prudnikov1990integrals}, the proof is accomplished.

Using \cite[eq. (1.2.4)]{mathai1978h}, the PDF of $\xi_k$ in (\ref{PDF_best_xi}) can be re-expressed in terms of Fox's $H$-function, 
\begin{equation} \label{PDF_xi_Fox}
f_{\xi_k}(x) = \frac{\delta}{x\Gamma(k)} H_{0,1}^{1,0} \left[ { A_{b1} x^\delta  \left| {\begin{array}{*{20}c}
    {-}  \\
   {(k,1)}  \\
\end{array}} \right.} \right],
\end{equation}
subsequently, using \cite[eq. (1.2.2)]{mathai1978h} of $\log_2(1+\frac{1}{x})$ and plugging (\ref{PDF_xi_Fox}), yields
\begin{equation}
\begin{split}
R_{B,k}^M & = \mathbb{E}_{\xi_k}\left[ \log_2\left(1 + \frac{\eta_k }{\xi_k} \right)\right] \\
& =\frac{\delta}{ \Gamma(k) \ln 2 } \int_0^\infty y^{-1} H_{2,2}^{1,2} \left[ { \frac{y}{\eta_k }  \left|\hspace{-1ex} {\begin{array}{*{20}c}
    {(1,1),(1,1)}  \\
   {(1,1),(0,1)}  \\
\end{array}} \right.}\hspace{-1ex} \right] \\
&\quad \quad\times H_{0,1}^{1,0} \left[ { A_{b1} y^\delta  \left| {\begin{array}{*{20}c}
    {-}  \\
   {(k,1)}  \\
\end{array}} \right.} \right] dy,
\end{split}
\end{equation}
next, using \cite[eq. (2.25.1.1)]{prudnikov1990integrals} and \cite[eq. (1.7.1)]{mathai1978h}, the proof is achieved.
\section*{Acknowledgment}
This work has been supported by the ÉTS research chair of physical layer security in wireless networks. The authors would like to thank Prof. Daniel B. da Costa for his helpful suggestions.

\ifCLASSOPTIONcaptionsoff
  \newpage
\fi

\bibliographystyle{IEEEtran}
\bibliography{AlphaMu}

\end{document}